%% file: ms.tex
\documentclass[aoas,preprint]{imsart}

\RequirePackage[OT1]{fontenc}
\RequirePackage{natbib}

\usepackage[a4paper, top=3cm, bottom=3cm]{geometry}
\usepackage[latin1]{inputenc}

\usepackage{tocloft}
\usepackage[usenames]{color}
\usepackage{amsmath,amsfonts,amssymb,amsthm}
\usepackage{graphicx}

\usepackage{paralist}

\usepackage{subfig}

\usepackage{amsbsy}

\usepackage[mathscr]{eucal}

\usepackage{bm}

\usepackage{xspace}

\usepackage{algorithm,algpseudocode,ifpdf}

\mathchardef\mhyphen="2D

\ifpdf
\usepackage[colorlinks,urlcolor=blue,citecolor=blue,linkcolor=blue]{hyperref}
\else
\newcommand{\href}[2]{{#2}}
\usepackage{url}
\fi

\newcommand{\argmin}[1]{\underset{#1}{\operatorname{argmin}}\text{ }}

\newcommand{\median}[1]{\underset{#1}{\operatorname{median}}\text{ }}

\ifpdf
\newcommand{\Sec}[1]{\hyperref[sec:#1]{\S\ref*{sec:#1}}} 
\newcommand{\App}[1]{\hyperref[sec:#1]{Appendix~\ref*{sec:#1}}} 
\newcommand{\Eqn}[1]{\hyperref[eq:#1]{{\rm (\ref*{eq:#1})}}} 
\newcommand{\Part}[1]{\hyperref[part:#1]{(\ref*{part:#1})}} 
\newcommand{\Fig}[1]{\hyperref[fig:#1]{Figure~\ref*{fig:#1}}} 
\newcommand{\Tab}[1]{\hyperref[tab:#1]{Table~\ref*{tab:#1}}} 
\newcommand{\Thm}[1]{\hyperref[thm:#1]{Theorem~\ref*{thm:#1}}} 
\newcommand{\Lem}[1]{\hyperref[lem:#1]{Lemma~\ref*{lem:#1}}} 
\newcommand{\Prop}[1]{\hyperref[prop:#1]{Proposition~\ref*{prop:#1}}} 
\newcommand{\Cor}[1]{\hyperref[cor:#1]{Corollary~\ref*{cor:#1}}} 
\newcommand{\Def}[1]{\hyperref[def:#1]{Definition~\ref*{def:#1}}} 
\newcommand{\Alg}[1]{\hyperref[alg:#1]{Algorithm~\ref*{alg:#1}}} 
\newcommand{\Ex}[1]{\hyperref[ex:#1]{Example~\ref*{ex:#1}}} 
\newcommand{\As}[1]{\hyperref[as:#1]{Assumption~{\rm\ref*{as:#1}}}} 
\newcommand{\Reg}[1]{\hyperref[as:#1]{Condition~\ref*{reg:#1}}} 
\newcommand{\AlgLine}[2]{\hyperref[alg:#1]{line~\ref*{line:#2} of Algorithm~\ref*{alg:#1}}}
\newcommand{\AlgLines}[3]{\hyperref[alg:#1]{lines~\ref*{line:#2}--\ref*{line:#3} of Algorithm~\ref*{alg:#1}}}
\else
\newcommand{\Sec}[1]{{\S\ref{sec:#1}}} 
\newcommand{\App}[1]{{Appendix~\ref{sec:#1}}} 
\newcommand{\Eqn}[1]{{(\ref{eq:#1})}} 
\newcommand{\Part}[1]{{(\ref{part:#1})}} 
\newcommand{\Fig}[1]{{Figure~\ref{fig:#1}}} 
\newcommand{\Tab}[1]{{Table~\ref{tab:#1}}} 
\newcommand{\Thm}[1]{{Theorem~\ref{thm:#1}}} 
\newcommand{\Lem}[1]{{Lemma~\ref{lem:#1}}} 
\newcommand{\Prop}[1]{{Proposition~\ref{prop:#1}}} 
\newcommand{\Cor}[1]{{Corollary~\ref{cor:#1}}} 
\newcommand{\Def}[1]{{Definition~\ref{def:#1}}} 
\newcommand{\Alg}[1]{{Algorithm~\ref{alg:#1}}} 
\newcommand{\Ex}[1]{{Example~\ref{ex:#1}}} 
\newcommand{\Reg}[1]{{R~\ref*{reg:#1}}} 
\fi

\newcommand{\Real}{\mathbb{R}}

\newcommand{\Tra}{^{\sf T}} 
\newcommand{\Inv}{^{-1}} 

\newcommand{\amp}{\mathop{\:\:\,}\nolimits}

\newcommand{\V}[1]{{\bm{\mathbf{\MakeLowercase{#1}}}}} 
\newcommand{\VE}[2]{\MakeLowercase{#1}_{#2}} 
\newcommand{\Vtilde}[1]{{\bm{\tilde \mathbf{\MakeLowercase{#1}}}}} 
\newcommand{\Vn}[2]{\V{#1}^{(#2)}} 
\newcommand{\VnE}[3]{{#1}^{(#2)}_{#3}} 

\newcommand{\M}[1]{{\bm{\mathbf{\MakeUppercase{#1}}}}} 
\newcommand{\ME}[2]{\MakeLowercase{#1}_{#2}} 

\arxiv{arXiv:0000.0000}

\startlocaldefs
\numberwithin{equation}{section}
\theoremstyle{plain}
\newtheorem{theorem}{Theorem}[section]
\newtheorem{proposition}{Proposition}[section]
\endlocaldefs

\begin{document}

\begin{frontmatter}
\title{Estimating a Common Period for a Set of Irregularly Sampled Functions with Applications to Periodic Variable Star Data}
\runtitle{Estimating a Common Period}

\begin{aug}
\author{\fnms{James P.} \snm{Long}\thanksref{m1}\ead[label=e1]{jlong@stat.tamu.edu}},
\author{\fnms{Eric C.} \snm{Chi}\thanksref{m2}\ead[label=e2]{echi@rice.edu}}
\and
\author{\fnms{Richard G.} \snm{Baraniuk}\thanksref{m2}
\ead[label=e3]{richb@rice.edu}}


\runauthor{J. P. Long et al.}

\affiliation{Texas A\&M University\thanksmark{m1} and Rice University\thanksmark{m2}}

\address{J. P. Long\\
Department of Statistics\\
Texas A\&M University\\
College Station, TX 77843\\
USA\\
\printead{e1}}

\address{E. C. Chi\\
Department of Electrical \\
and Computer Engineering\\
Rice University\\
Houston, TX 77005\\
USA \\
\printead{e2}}

\address{R. G. Baraniuk\\
Department of Electrical\\
and Computer Engineering\\
Rice University\\
Houston, TX 77005\\
USA \\
\printead{e3}}
\end{aug}

\begin{abstract}
We consider the estimation of a common period for a set of functions sampled at irregular intervals. The problem arises in astronomy, where the functions represent a star's brightness observed over time through different photometric filters. While current methods can estimate periods accurately provided that the brightness is well--sampled in at least one filter, there are no existing methods that can provide accurate estimates when no brightness function is well--sampled. In this paper we introduce two new methods for period estimation when brightnesses are poorly--sampled in all filters. The first, multiband generalized Lomb-Scargle (MGLS), extends the frequently used Lomb-Scargle method in a way that na\"{i}vely combines information across filters. The second, penalized generalized Lomb-Scargle (PGLS),  builds on the first by more intelligently borrowing strength across filters. Specifically, we incorporate constraints on the phases and amplitudes across the different functions using a non--convex penalized likelihood function. We develop a fast algorithm to optimize the penalized likelihood by combining block coordinate descent with the majorization-minimization (MM) principle. We illustrate our methods on synthetic and real astronomy data. Both advance the state-of-the-art in period estimation; however, PGLS significantly outperforms MGLS when all functions are extremely poorly--sampled.
\end{abstract}



\end{frontmatter}

\section{Introduction}
\label{sec:introduction}

Periodic variable stars play an important role in several areas of modern astronomy, including extragalactic distance determinations and estimation of the Hubble constant \citep{shappee2011new,riess20113}. To effectively use periodic variables, astronomers need accurate period estimates. For instance, in classification studies, periodic variables are assigned labels representing the underlying astrophysical reason for brightness variation. Period is one of the most useful features for determining a star's class, and incorrect period estimates are a leading cause of misclassifications \citep{richards2011machine,dubath2011random}. 

Astronomers estimate periods using the {\em light curve} of a star. A light curve is a set of brightness measurements of a star taken over time. Many astronomical surveys measure the brightness of stars in several photometric filters, or bands. Multiband data is useful because differences in brightness across bands is a strong indicator of stellar class. \Fig{lc_unfold} displays the light curve of the periodic variable \texttt{OGLE-LMC-T2CEP-041} observed in the I band (orange $\times$) and V band (blue $\circ$) by the Optical Graviational Lensing Experiment (OGLE) \citep{udalski2008optical}. This star has been observed 702 times in the I band and 76 times in the V band over the course of roughly 4000 days. The intervals between brightness measurements are irregular and the star is observed at different times in the different bands. This is typical for light curves. Many stars are behind the sun for part of the year, leading to months long gaps between observations. Additionally, weather can disrupt planned observation times. The vertical bars around each point (not always visible due to being very small) are one-standard deviation uncertainty measurements on brightness. The size of the error bars relative to the amount of variation in brightness demonstrate that \texttt{OGLE-LMC-T2CEP-041} is a variable star.

\texttt{OGLE-LMC-T2CEP-041} is a periodic variable with period of approximately 2.48 days. The pattern of variation in \texttt{OGLE-LMC-T2CEP-041} can be observed by plotting the brightness of the star versus phase (time modulo period). This is known as the {\em folded} or {\em phased light curve}. \Fig{lc_fold} 
displays the folded light curve in each band.

\begin{figure}[t]
\begin{center}
\centering
 \subfloat[Light Curve]{\label{fig:lc_unfold}
 \includegraphics[scale=.4]{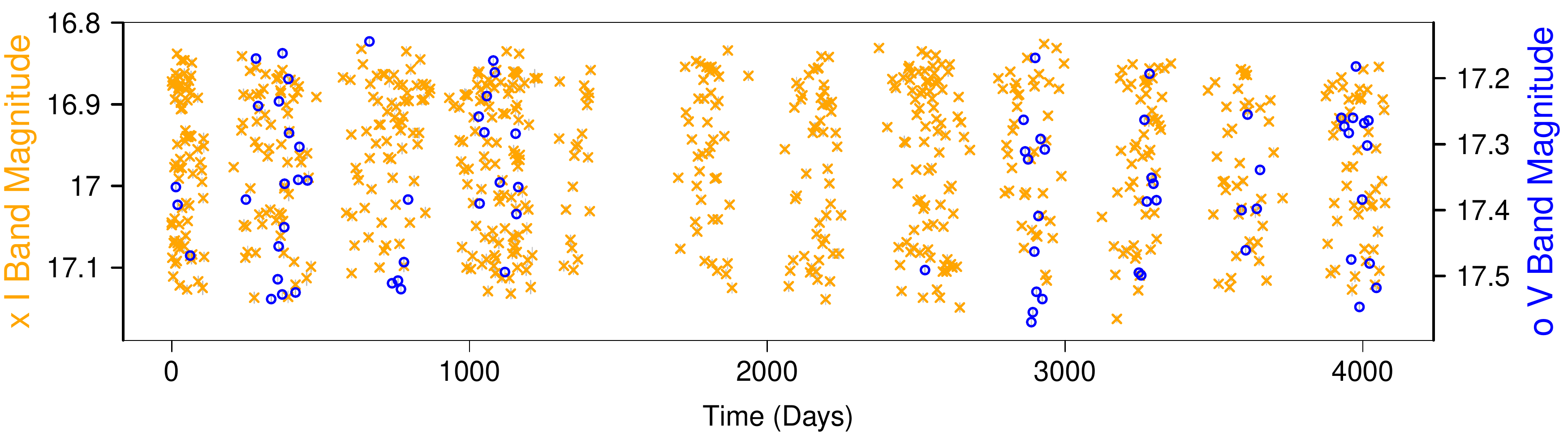}
 } \\
 \subfloat[Folded Light Curve]{\label{fig:lc_fold}
 \includegraphics[scale=.4]{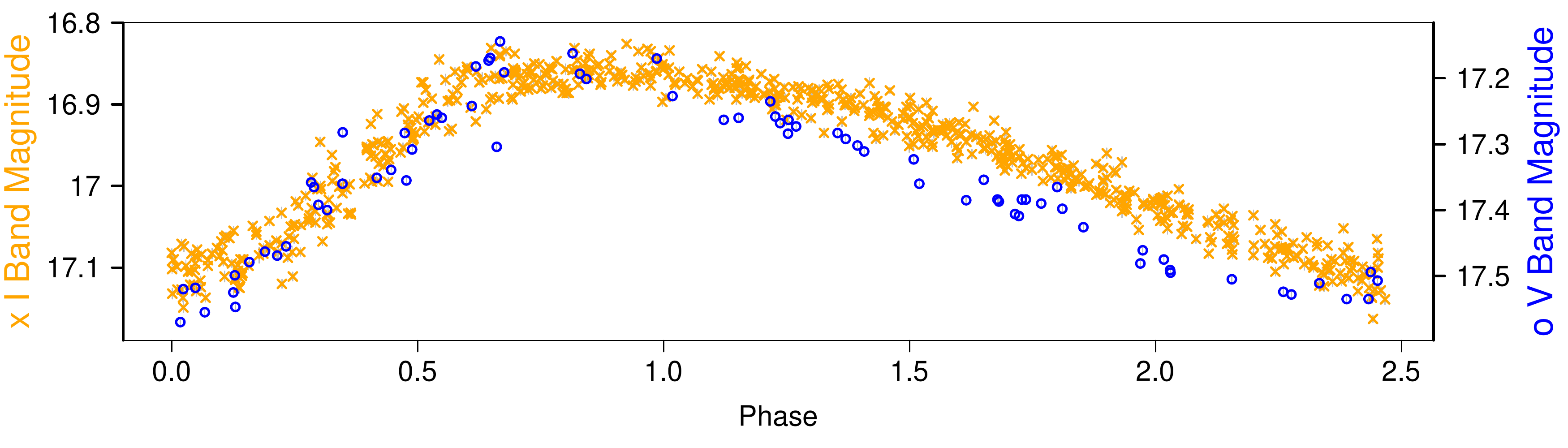}
 }
 \caption{(a) The light curve of the star \texttt{OGLE-LMC-T2CEP-041} in the I band (orange $\times$) and V band (blue $\circ$). Given these measurements, astronomers seek to estimate the common period (in this case 2.48 days) of the I and V bands. 
(b) The pattern of brightness variation becomes apparent when we plot brightness (i.e. magnitude) versus time modulo period. This is known as the folded light curve. Note, the phase of maximum brightness in the I band and V band are similar. In this work we propose using correlations between phase of maximum brightness in different bands in order to improve period estimation.}
\end{center}
\end{figure}

Accurate estimation of \texttt{OGLE-LMC-T2CEP-041}'s period is fairly easy, because the star has been observed hundreds of times in the I band. However some astronomical surveys collect many fewer measurements per band. For example, in Stripe 82, the Sloan Digital Sky Survey I (SDSS--I) collected light curves for $\sim 700,000$ stars in five bands with a median of 10 observations per band. \cite{sesar2007exploring} identified several hundred of these stars as candidate periodic variables belonging to the class RR Lyrae but could not estimate periods due to the lack of available methodology for estimating periods with poorly--sampled light curves. Later, the Sloan Digital Sky Survey II (SDSS--II) roughly tripled the number of observations per light curve to a median of 30 observations per band. In a follow-up study, \cite{sesar2010light} used this expanded set of data to estimate periods. More recently, the PanSTARRS1 survey has collected five band variable star data with a median of 4 observations per band \citep{schlafly2012photometric}. In this article, we develop methology for estimating periods with this quality of data. Historical data consisting of well--observed light curves, such as SDSS-II, enables testing of algorithms on real data.



Period estimation is challenging when stars are poorly observed in multiple bands, because one must model brightness variation for several functions, thus using many degrees of freedom. For example, consider a model with $r$ parameters to describe the shape of the light curve in each band. If one collects data in $5$ bands with $10$ observations per band, then the model will have $5r + 1$ parameters ($1$ for the period), constrained by a total of 50 observations. In contrast with 50 observations all taken in a single band, there are only $r + 1$ parameters to fit.

Light curve shapes across different bands, however, are not independent. For example, periodic variables typically reach their peak brightness at similar points in phase space in each band (see \Fig{lc_fold}). Enforcing such physical constraints in models may improve the accuracy of period estimation procedures by reducing the effective degrees of freedom that must be used to fit the curves. We now present an illustrative example that confirms this intuition and motivates our strategy for borrowing strength across multiple photometric bands.

As a test, we use a sample of well--observed ($\geq$ 50 measurements / band) periodic variable stars observed in the I and V bands collected by the OGLE survey \citep{udalski2008optical}. We estimate their periods using a simple multiband extension of Generalized Lomb--Scargle (GLS). GLS models the brightness variation as sinusoidal and finds the best fitting period (see Section \ref{sec:period_est} for details). These periods are nearly correct for every star, because the light curves are well--observed. \Fig{phase_ests_ls_rich}
 shows the GLS phase estimates in the I and V bands. Phase is the location (here measured on a $[-\pi,\pi)$ scale) in the phased light curve where the brightness reaches a maximum.\footnote{A light curve with a phase of $-\pi$ is brightest at time modulo period equals 0. A light curve with a phase of $0$ is brightest at time modulo period equals half the period.} From the plot, it is evident that there is a strong correlation between the I and V band phases. This was also evident for \texttt{OGLE-LMC-T2CEP-041} in \Fig{lc_fold}, 
   where both I and V bands peaked near phase $-\pi/3$.

Now consider downsampling these light curves to 10 measurements in both the I and V bands. The data now resemble the quality, in terms of number of photometric measurements per band, of SDSS-I. We again apply the GLS algorithm and plot the phase estimates in \Fig{phase_ests_ls_poor}. 
 The points are colored and marked by whether the period is estimated to within 1\% of its true value. By this measure, \input{figs/cep_gls_tab} \unskip of the periods are estimated incorrectly. Note that the phase estimates do not appear strongly correlated in the two bands. The GLS algorithm does not force the phases to be similar in different bands. An algorithm that uses the known phase correlations may be able to estimate periods more accurately. In \Fig{phase_ests_bcd_poor}
  we plot the phase estimates using a modified GLS algorithm, termed Penalized GLS (PGLS), that we develop in Section \ref{sec:est}. PGLS enforces phases constraints across the bands using a penalized likelihood. The result is tightly correlated phase estimates that are physically realistic. Importantly, incorrect period estimates have fallen to \input{figs/cep_pgls_tab} \unskip from \input{figs/cep_gls_tab}\unskip.

\begin{figure}[t]
\centering
 \subfloat[GLS on well--observed.]{\label{fig:phase_ests_ls_rich}
\includegraphics[scale=.3]{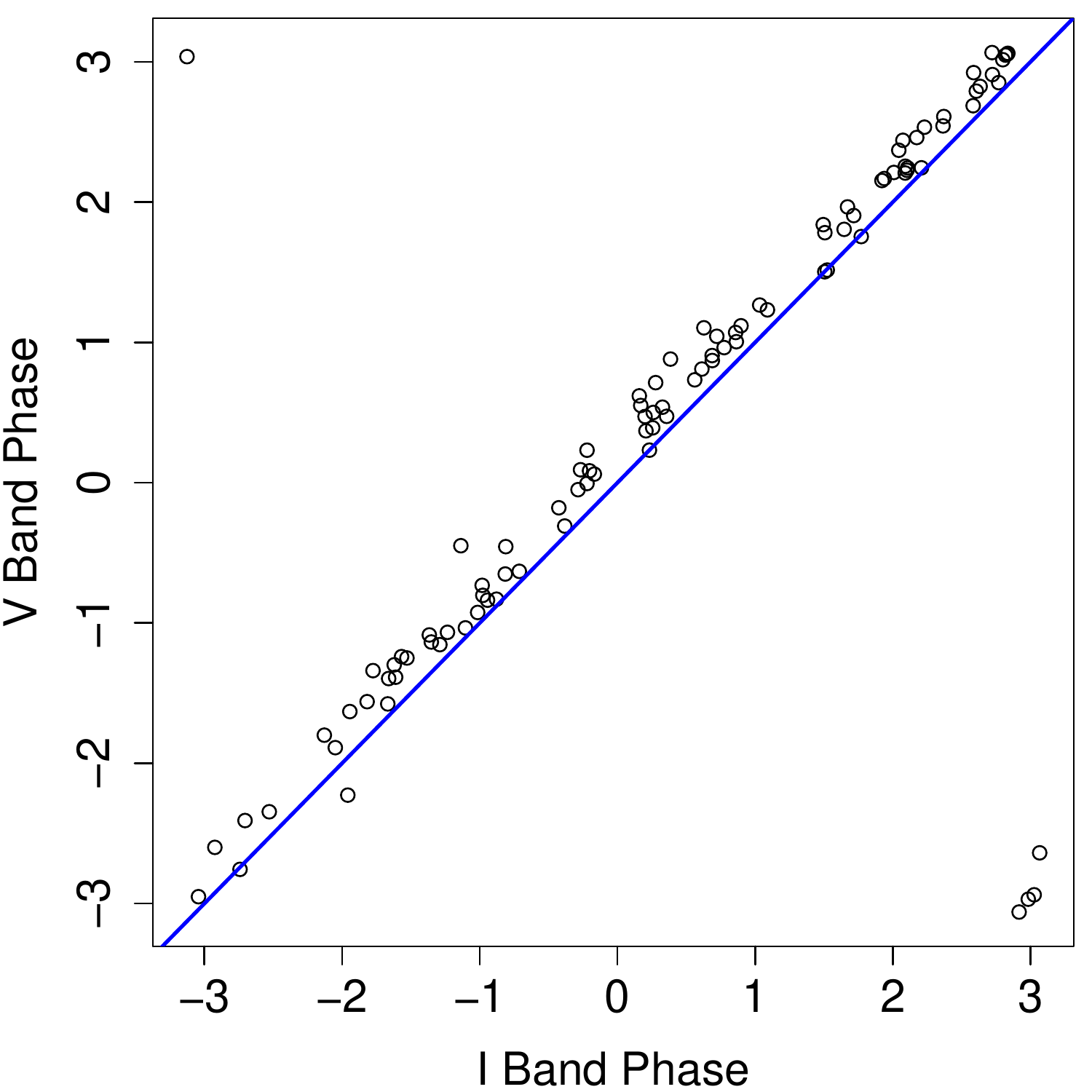}
 } 
 \subfloat[GLS on poorly observed.]{\label{fig:phase_ests_ls_poor}
\includegraphics[scale=.3]{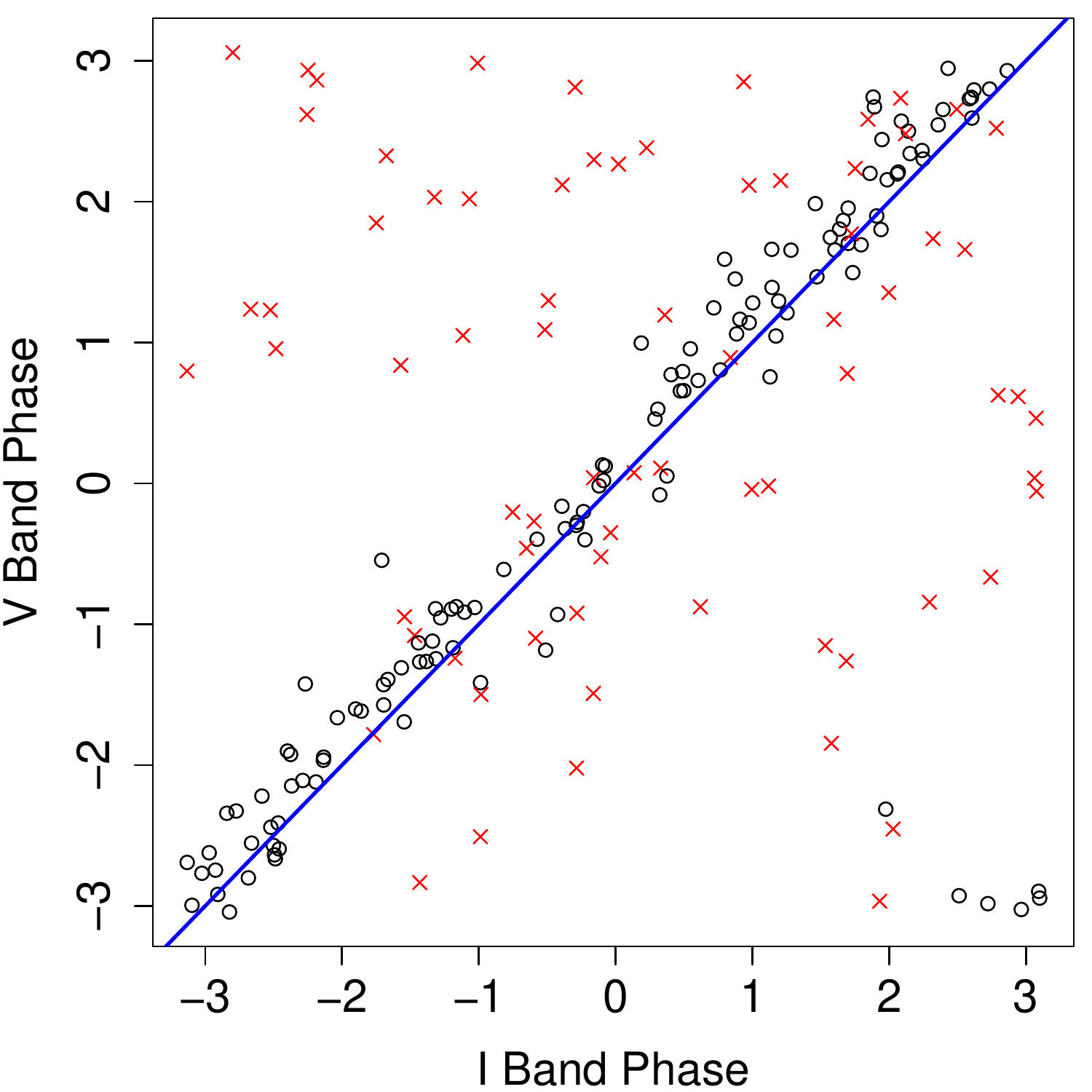}
}
 \subfloat[PGLS on poorly observed.]{\label{fig:phase_ests_bcd_poor}
\includegraphics[scale=.3]{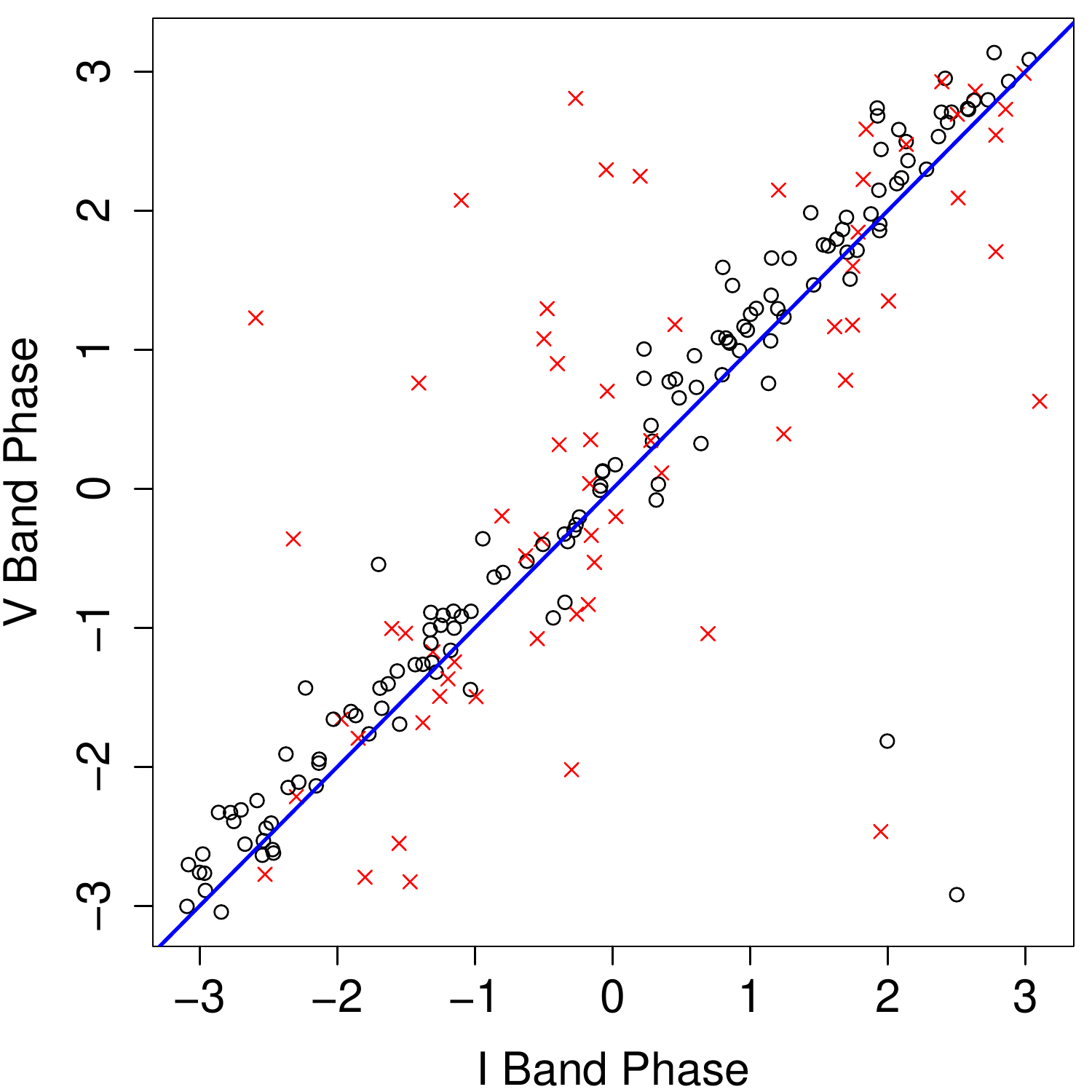}
}
 \caption{(a) Phase estimates using GLS period estimates on light curves with many observations in both I and V bands. (b) When the light curves are downsampled to 10 observations per band, GLS produces phase estimates that are only weakly correlated. Furthermore, \protect\input{figs/cep_gls_tab} \unskip of periods are estimated incorrectly (red $\times$). (c) Using the same data as in \ref{fig:phase_ests_ls_poor}, PGLS produces more correlated phase estimates. Incorrect period estimates drop to \protect\input{figs/cep_pgls_tab}\unskip.}
\end{figure}

The remainder of this paper is structured as follows: In Section \ref{sec:back_related} we discuss additional astronomy background and existing methodology for estimating periods of variable stars. In Section \ref{sec:est} we introduce a penalized likelihood method for estimating periods that uses known correlations in phase and amplitude across bands. In Section \ref{sec:algorithm} we develop a fast algorithm to maximize the likelihood by combining block coordinate descent with the majorization-minimization (MM) principle. In Section \ref{sec:gamma_selection} we discuss selection of tuning parameters. We apply our algorithm to simulated and real data in Section \ref{sec:data}. We finish with conclusions in Section \ref{sec:conclusions}.

\section{Astronomy Background and Existing Period Estimation Methods}
\label{sec:back_related}

\subsection{Photometry}

Telescopes take images of the sky using a particular filter, also known as passband or band. From these images, stars are identified and their brightnesses estimated. Let $m_{bi}$ be the brightness of a star observed in band $b$ for an image taken at time $t_{bi}$. The standard unit of brightness in astronomy is magnitude. Magnitude is a negative log transform of luminosity, thus brighter objects have lower magnitudes. An uncertainty on $m_{bi}$, denoted $\sigma_{bi}$, is also typically recorded. Often the observed magnitude $m_{bi}$ is modeled as being some true, unobserved magnitude plus Gaussian error with standard deviation $\sigma_{bi}$.  Using this notation, the light curve for a particular star is $\{\{(t_{bi},m_{bi},\sigma_{bi})\}_{i=1}^{n_b}\}_{b=1}^B$ where $B$ is the number of bands and $n_b$ is the number of images of the star taken in band $b$. Modern astronomy surveys collect millions of light curves.

\subsection{Identifying Variable Stars}
 
Before period estimation is performed, astronomers typically separate the non--variable and variable stars in a survey using variability detection methods. Simple methods involve computing the difference between the magnitudes and the mean magnitude for each band. Large absolute differences suggest the star is a variable. This approach does not use any time information. The Welch--Stetson statistics developed in \cite{welch1993robust} and \cite{stetson1996automatic} use time correlations in residuals. These statistics are designed for multiband data where images are taken nearly simultaneously in different bands. 


\subsection{Period Estimation}
\label{sec:period_est}
From the set of variable stars, the periodic variables are identified and characterized. Astronomers and statisticians have developed many methods for estimating periods. Nearly all of the existing methodology has been designed for single band data observed many times. We refer interested readers to \cite{reimann1994frequency} Chapter 3.1 and \cite{graham2013comparison} for reviews of existing methodology.

Here we describe a pair of approaches, known in astronomy as Lomb--Scargle (LS) and generalized Lomb--Scargle (GLS). The multiband period estimation method developed in this article is an extension of GLS. LS and GLS both model the magnitudes as a sinusoid plus measurement error. The period estimate is the period that maximizes the likelihood. With LS the magnitudes are scaled to mean 0, and a sinusoid without an intercept is fit to the data  \citep{lomb1976least,scargle1982studies}. \cite{zechmeister2009generalised} proposed a modification of LS, GLS, in which the magnitudes are not normalized to mean $0$ and an intercept term is used in the sine fit. Specifically let
\begin{eqnarray*}
m_i & = & \beta_0 + a \sin(\omega t_i + \rho) + \epsilon_i,
\end{eqnarray*}
where $\beta_0$, $a$, $\rho$, and $\omega$ are the magnitude offset, amplitude, phase, and frequency of the star. The measurement error $\epsilon_i \sim N(0,\sigma_i^2)$ is assumed to be independent across $i$. Although the model is non-linear in the frequency $\omega$, one can nonetheless compute maximum likelihood estimates by selecting a grid $\Omega$ of possible frequencies $\omega$ and solving a linear regression at each frequency in the grid. We describe this procedure in Section \ref{sec:grid_unpenalized}.

The sinusoidal model is an approximation, as most light curves exhibit at least some degree of non--sinusoidal variation. However, LS and GLS have proven to be effective period estimation algorithms in many settings. \cite{reimann1994frequency} (Section 3.3) found that for periodic light curves with unimodal behavior over phase, LS estimated frequency as well as more sophisticated approaches, including sinusoidal models with harmonic terms, periodic cubic splines, and SuperSmoother.\footnote{See \cite{friedman1984variable} for a description of SuperSmoother.} In a study of periodic variables with a variety of light curve shapes, \cite{dubath2011random} found that GLS outperformed competing period estimation methods. Part of the success of LS and GLS is due to being computationally fast relative to many alternatives. LS and GLS have been used in many recent studies of periodic variable stars \citep{richards2011machine,debosscher2009automated,watkins2009substructure,suveges2012search}.

\subsection{Multiband Period Estimation Methods} 

While many period estimation procedures exist for single band data, few methods have been developed for multiband light curves. With multiband data, practitioners typically use single band period estimation procedures individually on each band and then combine or choose between the single band estimates based on some criteria. For example \cite{watkins2009substructure} used the LS sinusoid model run separately on g and r band light curves to select candidate periods that were then analyzed using a string--length period algorithm to determine a single best estimate.

\cite{suveges2012search} applied GLS to multiband data by first combining the bands together using a robust version of principal components analysis. Their method assumes that brightness measurements in the different bands are taken at the same time, limiting its applicability. The two methods we introduce below do not require simultaneous measurements.




\section{A Penalized Maximum Likelihood Multiband Period Estimator}
\label{sec:est}

\subsection{Multiband Lomb Scargle}

We begin by proposing a simple generalization of the single--band GLS algorithm of \cite{zechmeister2009generalised} to multiple bands. We call this first method MGLS for multiband GLS. To our knowledge MGLS is the first work to formally define GLS for multiple bands.

Throughout, scalars are denoted by lowercase letters ($u$), vectors by boldface lowercase letters ($\V{u}$), and matrices by boldface capital letters ($\M{U}$). We denote the standard dot product between vectors $\V{a}$ and $\V{b}$ by $\langle \V{a}, \V{b} \rangle = \V{a}\Tra\V{b}$. We denote the vector of all ones with $\V{1}$.

Let $D = \{\{(t_{bi},m_{bi},\sigma_{bi})\}_{i=1}^{n_b}\}_{b=1}^B$ denote the data, namely the triples of times, magnitudes, and magnitude error measurements in the $B$ bands. Let $\omega$ be the frequency that is shared across all bands. Brightness at time $t$ in band $b$ is modeled using the sine curve
\begin{eqnarray*}
\mu_b(t) & = & \VE{a}{b} \sin(\omega t + \rho_b) + \beta_{0b}
\end{eqnarray*}
where $\VE{a}{b}$, $\rho_b$ and $\beta_{0b}$ are the amplitude, phase, and magnitude offset for band $b$. 
For notational efficiency we refer to vectors of amplitudes, phases, and magnitude offsets as
\begin{eqnarray*}
\V{A} & = & (a_1, \ldots, a_B)\Tra\\
\V{\rho} & = & (\rho_1,\ldots,\rho_B)\Tra\\
\V{\beta_0} & = & (\beta_{01},\ldots,\beta_{0B})\Tra.
\end{eqnarray*}
We bundle these parameters into the vector $\V{\theta} = (\omega,\V{A}\Tra,\V{\rho}\Tra,\V{\beta_0}\Tra) \in \Real^{3B + 1}$.
The observed magnitude $m_{bi}$ is a noisy measurement of $\mu_b$ at time $t_{bi}$, namely
\begin{eqnarray*}
m_{bi} & = & \mu_b(t_{bi}) + \epsilon_{bi},
\end{eqnarray*}
where $\epsilon_{bi} \sim N(0,\sigma_{bi}^2)$. Assuming mutual independence of all $\epsilon_{bi}$, the likelihood function becomes
\begin{eqnarray}
\label{eq:like}
p(D \mid \V{\theta}) & = & \prod_{b=1}^B \prod_{i=1}^{n_b} \frac{1}{\sqrt{2\pi}\sigma_{bi}} e^{-[m_{bi} - \mu_b(t_{bi})]^2 / (2\sigma_{bi}^2)}.
\end{eqnarray}
Note that to ensure identifiability of this model we require that the elements $\VE{a}{b}$ be nonnegative and the elements $\VE{\rho}{b}$ reside in the interval $[0,\pi)$. To keep our notation readable, we do not explicitly write out these constraints in the rest of the paper. 

The negative log likelihood (NLL) of $p(D \mid \V{\theta})$ is
\begin{eqnarray}
\ell(\omega,\V{\beta_0},\V{a},\V{\rho} \mid D) & = & \frac{1}{2} \sum_{b=1}^B \sum_{i=1}^{n_b} \left(\frac{m_{bi} - A_b\sin(t_{bi}\omega + \rho_b) - \beta_{b0}}{\sigma_{bi}}\right)^2
\end{eqnarray}
The solution to the NLL at frequency $\omega$ is
\begin{eqnarray*}
\ell(\omega) & = & \min_{\V{a},\V{\rho},\V{\beta_0}} \ell(\omega,\V{\beta_0},\V{a},\V{\rho})
\end{eqnarray*}
and the maximum likelihood estimate for $\omega$ is
\begin{eqnarray}
\label{eq:simple_multi_gls}
\widehat{\omega} & = & \argmin{\omega \in \Omega} \ell(\omega).
\end{eqnarray}
Equation \ref{eq:simple_multi_gls} is a straightforward generalization of the GLS algorithm to multiple bands. We discuss the computation of $\widehat{\omega}$ in Section \ref{sec:grid_unpenalized}.

\subsection{Penalty Terms}
\label{sec:penalty}
We now propose a generalization of GLS which penalizes unlikely values of $\V{a}$ and $\V{\rho}$. We call this method PGLS for penalized generalized Lomb-Scargle. The penalized negative log likelihood (PNLL) is
\begin{eqnarray}
\label{eq:penlike}
  f(\omega,\V{\beta_0},\V{a},\V{\rho} \mid D;\gamma_1,\gamma_2) & = & \ell(\omega,\V{\beta_0},\V{a},\V{\rho} \mid D)  + \gamma_1 J_1(\V{a}) + \gamma_2 J_2(\V{\rho}).
\end{eqnarray}
The minimal PNLL at frequency $\omega$ is
\begin{eqnarray}
\label{eq:pnll}
f(\omega \mid D;\gamma_1,\gamma_2) & = &   \min_{\V{\beta_0},\V{a},\V{\rho}} f(\omega,\V{\beta_0},\V{a},\V{\rho} \mid D;\gamma_1,\gamma_2),
\end{eqnarray}
and the PGLS frequency estimate is
\begin{eqnarray}
\label{eq:penalty_multi_gls}
\widehat{\omega}(D;\gamma_1,\gamma_2) & = & \argmin{\omega \in \Omega} f(\omega \mid D;\gamma_1,\gamma_2).
\end{eqnarray}
Note that we often suppress dependence of $f$ and $\widehat{\omega}$ on $D$, $\gamma_1$ and $\gamma_2$ when these quantities are fixed. The penalties $J_1$ and $J_2$ (see \eqref{eq:amp_penalty} and \eqref{eq:rho_penalty}) are chosen to be large for values of $\V{a}$ and $\V{\rho}$ that are unlikely.

We motivate the form of $J_1$ and $J_2$ using periodic variable star data from the OGLE survey \citep{udalski2008optical}. We fit MGLS to 100 well--observed stars in the I and V bands for four classes of periodic variables (Type I Cepheid, Type II Cepheid, RR Lyrae AB, RR Lyrae C). Since these stars have been well--observed in the I and V bands (at least 50 measurements / band), the parameter estimates should be quite accurate. 
\begin{figure}[t]
  \begin{center}
    \begin{includegraphics}[height=4in,width=4in]{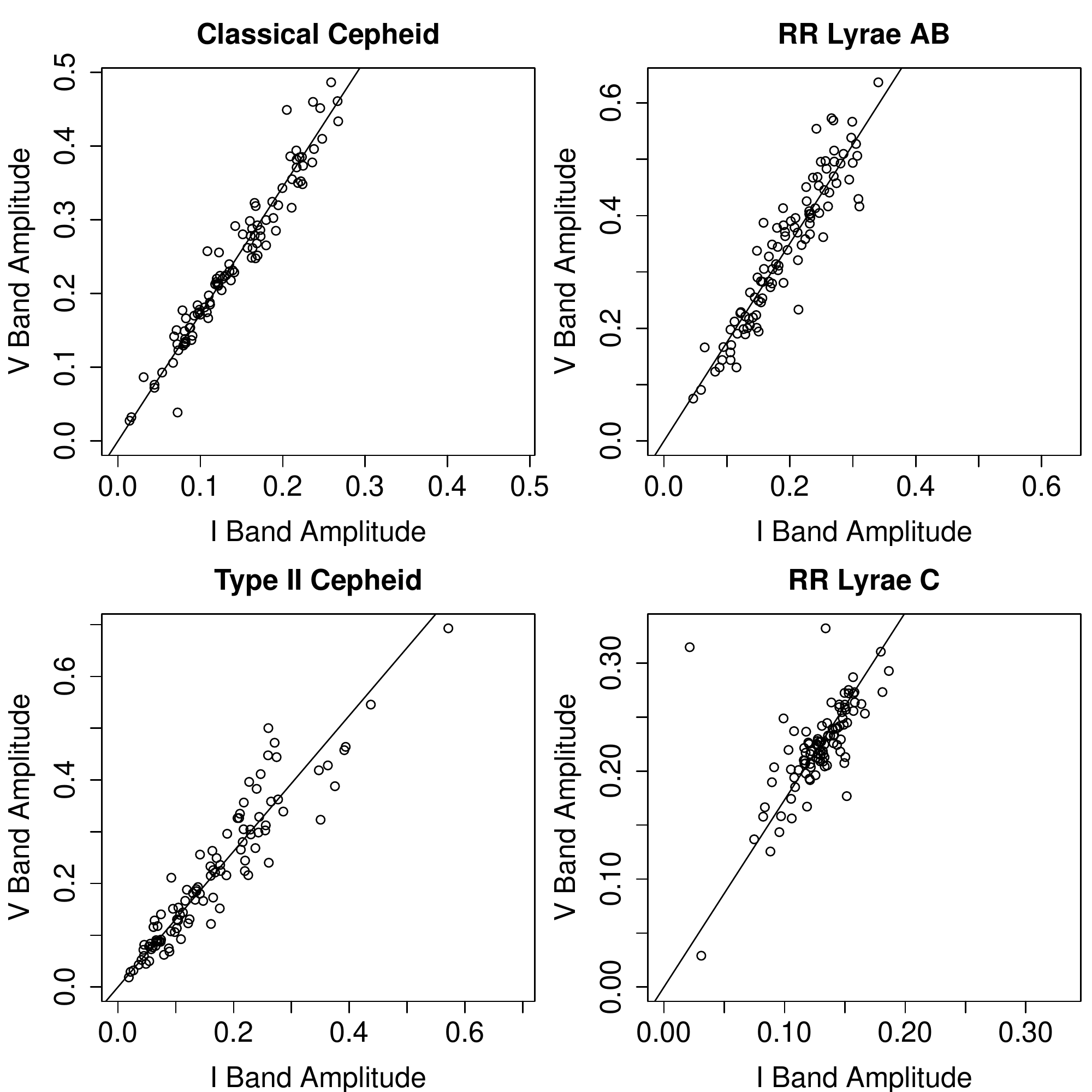}
      \caption{Amplitude correlations in the I and V bands for classes of periodic variable stars.\label{fig:scatter_A}}
    \end{includegraphics}
  \end{center}
\end{figure}
In Figure \ref{fig:scatter_A} we plot the maximum likelihood estimates of amplitude in the I and V bands for the four classes of periodic variables. In each plot, the line is chosen to pass through $(0,0)$ and the mean of the amplitudes in each band. Amplitudes cluster around these lines. The slope of these lines is greater for the RR Lyrae classes than for either Cepheid class. Let $\Vtilde{a}$ be the vector of mean amplitudes in each band. We propose penalizing the amplitude vector with
\begin{equation}
\label{eq:amp_penalty}
J_1(\V{a}) \amp = \amp \frac{1}{2} \V{a}\Tra \left (\M{I} - \Vtilde{a}{\Vtilde{a}}\Tra \right ) \V{a} \amp = \amp \frac{1}{2}\lVert\V{a} - \Vtilde{a}\Tra\V{a}\Vtilde{a}\rVert_2^2.\\
\end{equation}
where $\Vtilde{a}\Tra \Vtilde{a}=1$. In words, $J_1(\V{a})$ is half the squared norm of the amplitude component orthogonal to $\tilde{\V{a}}$. 
\begin{figure}[t]
  \begin{center}
    \begin{includegraphics}[height=4in,width=4in]{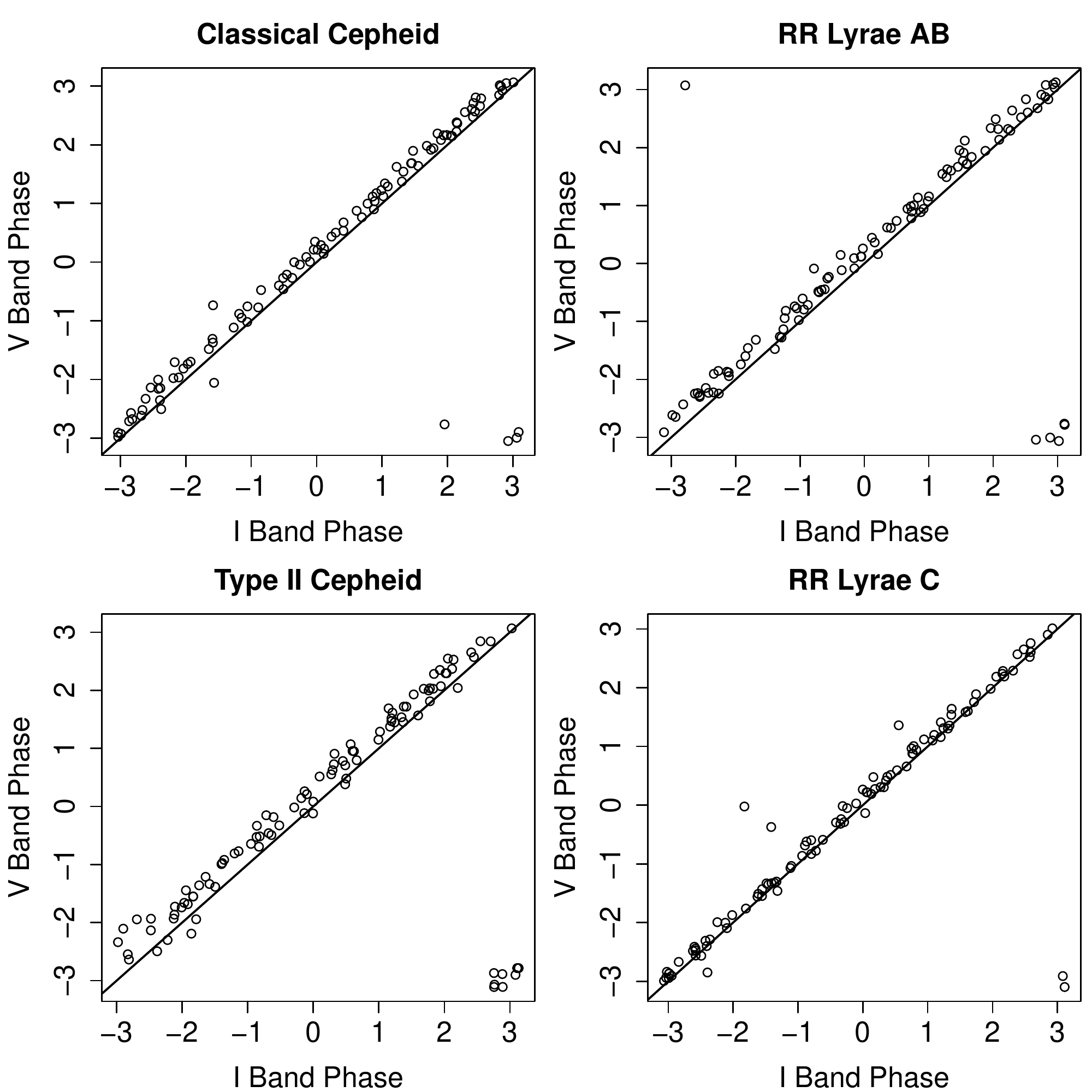}
      \caption{Phase correlations in the I and V bands for four classes of periodic variable stars.\label{fig:scatter_rho}}
    \end{includegraphics}
  \end{center}
\end{figure}

In Figure \ref{fig:scatter_rho} we plot the maximum likelihood estimates for I band and V band phases by class along with a line with slope 1 and y--intercept of 0. For each class, phase estimates usually are close to the vector $\V{1}$, implying that the light curve reaches maximum brightness in I and V bands at approximately the same point in phase space. Since a phase of $-\pi$ and a phase of $\pi$ are the same, the light curves represented by points in the upper left (near $(-\pi,\pi)$) and lower right (near $(\pi,-\pi)$) in each plot are not actually outliers.\footnote{Note that for the Type I and Type II Cepheid classes, the V band phase is slightly larger on average than the I band phase, implying the peak brightness occurs earlier in I band slightly earlier than V band. This has been documented in the astronomy literature; see Table 5 of \cite{freedman1988new}.} Based on these plots, $\V{\rho}$ is penalized with
\begin{equation}
\label{eq:rho_penalty}
J_2(\V{\rho}) \amp = \amp \frac{1}{2} \V{\rho}\Tra \left (\M{I} - \frac{1}{B}\V{1}\V{1}\Tra \right )\V{\rho} \amp = \amp \frac{1}{2} \lVert \V{\rho} - \frac{\V{1}\Tra \V{\rho}}{\V{1} \Tra \V{1}}\V{1}\rVert_2^2.
\end{equation}
In words, $J_2(\V{\rho})$ is half the squared norm of the component of $\V{\rho}$ that is orthogonal to $\V{1}$. The parameter $\gamma_2$ controls how strongly phase estimates are forced towards a multiple of the vector $\V{1}$. This parameter may be set based on the type of periodic variable. For example, if a practitioner is attempting to find periods for a class of variable stars where the phases in different bands are expected to be nearly identical, then $\gamma_2$ may be made very large. In contrast, if only a weak relationship is expected, $\gamma_2$ may be made much smaller. In Figure \ref{fig:scatter_rho}, there is roughly the same scatter for all classes, so a single $\gamma_2$ may be appropriate. We discuss selection of $\gamma_2$ further in Section \ref{sec:gamma_selection}. Note that with $\Vtilde{a} = (B^{-1/2},\ldots,B^{-1/2})$, $J_1$ has the same form as $J_2$.

In the remainder of this section, we discuss the existing literature on penalities similar to $J_1$ and $J_2$. To our knowledge there is no existing work on using these penalties for period estimation. An algorithm for computing $\widehat{\omega}$ in \ref{eq:penalty_multi_gls} is discussed in Section \ref{sec:algorithm}. Methodology for choosing $\gamma_1$ and $\gamma_2$ is discussed in Section \ref{sec:gamma_selection}.


\subsection{Discussion of the Penalties}
\label{sec:discuss_penalty}


Both penalties terms can be written generically as $J(\V{\nu}) = \V{\nu}\Tra \M{\Lambda}\ \V{\nu}$ where $\M{\Lambda}$ is positive semidefinite. When $\M{\Lambda}$ is the identity matrix, we recover the classic ridge penalty \citep{HoeKen1970}. When $\M{\Lambda}$ encodes a discrete differential operator, we recover the quadratic penalty broadly used in smoothing splines, reproducing kernel Hilbert spaces, and functional data analysis \citep{Ram2004}. In fact such quadratic penalties have garnered attention in functional principal components analysis \citep{HuaSheBuj2009, TiaHuaShe2012,All2013, AllGroTay2014} when there is a natural ``adjacency" among parameters and as well as a prior belief that variation among adjacent parameters should be smooth.

Here we propose using the discrete first order differential operator that corresponds to taking $\M{\Lambda}$ to be the graph Laplacian \citep{Chu1997} of a completely connected network of $B$ nodes where each node corresponds to a variable (amplitude or phase) of a given band. The graph Laplacian has been previously employed to enforce smoothness of regression parameters corresponding to 
neighbors in a gene network \citep{LiLi2008}. It has also been applied in enforcing smoothness in spatial parameters for estimating allele frequencies across populations in neighboring geographic regions
\citep{GelVou2003,RanNovLan2014}.

\section{Minimizing the Penalized Likelihood}
\label{sec:algorithm}

In this section we discuss computation of the PGLS frequency estimate $\widehat{\omega}$ in \eqref{eq:penalty_multi_gls}. The algorithm consists of:
\begin{enumerate}
\item Choosing a grid of possible frequencies $\Omega$.
\item Minimizing the negative log likelihood (NLL) for every frequency in the grid.
\item Minimizing the penalized negative log likelihood (PNLL) using a block coordinate descent algorithm for a sufficient subset of the frequencies in the grid.
\end{enumerate}

In Section \ref{sec:grid_unpenalized} we discuss steps 1 and 2. The frequency that minimizes the NLL is the MGLS frequency. In Section \ref{sec:necessary} we show that in order to obtain the PGLS frequency estimate, it is necessary only to minimize the PNLL on a subset of $\Omega$. In Section \ref{sec:bcd} we discuss step 3.

\subsection{Choice of Frequency Grid and Minimizing the Negative Log Likelihood}
\label{sec:grid_unpenalized}

The grid of frequencies $\Omega$ is typically chosen on a linear scale with endpoints representing the physical minimium and maximum frequencies possible for the periodic variables of interest. For example, RR Lyrae type stars are known to have periods ranging from 0.1 days to 1 day (Chapter 6.8 of \cite{percy2007understanding}). Therefore the grid endpoints for $\Omega$ will occur at $(2\pi)/1 \approx 6.28$ and $(2\pi)/0.1 = 62.8$. In contrast, Cepheids generally have periods ranging from 1 day to 100 days, and thus a grid of frequencies from $0.06$ to $6.28$ is appropriate (Chapter 6.9 of \cite{percy2007understanding}). In periodic variable classification studies where a wide set of periodic variable classes are of interest, periods may range from a few hours to hundreds of days \citep{richards2011machine}.

The grid spacing generally depends on the length of time between the first and last observation. As a star is observed for a longer and longer time, small errors in frequency result in larger and larger shifts in phase space for the folded light curve. \cite{richards2011machine} used a grid with frequencies spaced $0.1/(\text{max} - \text{min})$ where $\text{max}$ and $\text{min}$ are the maximum time and minimum times of all star observations. We use this spacing in all examples here. For the SDSS-II RR Lyrae light curves studied in Section \ref{sec:sdss}, there are around 20,000 frequencies in $\Omega$.


The MGLS frequency estimate
\begin{eqnarray*}
\widehat{\omega} & = & \argmin{\omega \in \Omega} \! \! \min_{\V{a},\V{\rho},\V{\beta_0}} \ell(\omega,\V{\beta_0},\V{a},\V{\rho})
\end{eqnarray*}
is determined by minimizing $\ell$ with respect to $\V{\beta_0},\V{a},\V{\rho}$ at every frequency in the grid. The function $\ell$ can be minimized with respect to $\V{\beta_0},\V{a},\V{\rho}$ by performing $B$ linear regressions. In particular, with the first equality by definition and the third equality by the angle addition formula
\begin{align*}
&\min_{\V{a},\V{\rho},\V{\beta_0}} \ell(\omega,\V{\beta_0},\V{a},\V{\rho})\\
&= \amp \min_{\V{a},\V{\rho},\V{\beta_0}} \frac{1}{2} \sum_{b=1}^B \sum_{i=1}^{n_b} \left(\frac{m_{bi} - A_b\sin(t_{bi}\omega + \rho_b) - \beta_{b0}}{\sigma_{bi}}\right)^2\\
&= \amp \frac{1}{2} \sum_{b=1}^B \min_{a_b,\rho_b,\beta_{0b}} \sum_{i=1}^{n_b} \left(\frac{m_{bi} - A_b\sin(t_{bi}\omega + \rho_b) - \beta_{b0}}{\sigma_{bi}}\right)^2\\
&= \amp \frac{1}{2} \sum_{b=1}^B \min_{a_b,\rho_b,\beta_{0b}} \sum_{i=1}^{n_b} \left(\frac{m_{bi} - A_b\cos(\rho_b)\sin(t_{bi}\omega) - A_b\sin(\rho_b)\cos(t_{bi}\omega) - \beta_{b0}}{\sigma_{bi}}\right)^2.
\end{align*}
The summand is the residual sum of squares when regressing $(m_{b1}, \ldots, m_{bn_b})$ on to the predictors $(\sin(t_{b1}\omega), \ldots, \sin(t_{bn_b}\omega))$ and $(\cos(t_{b1}\omega), \ldots, \cos(t_{bn_b}\omega))$ with known observation weights $\sigma_{b1}^{-2}, \ldots, \sigma_{bn_b}^{-2}$. 

\subsection{Minimizing the PNLL for a Necessary Subset of $\Omega$}
\label{sec:necessary}

Compared to minimizing the NLL at a fixed $\omega$, minimizing the PNLL is computationally expensive (see Section \ref{sec:bcd}). Thus we do not want to compute the PNLL at every element in a potentially large grid $\Omega$. Here we show that it is possible to only compute the PNLL on a subset of $\Omega$ and still be guaranteed of finding the frequency that minimizes the PNLL. The basic idea is to use the NLL $\ell$ as a pointwise lower bound of the PNLL $f$ in order to quickly check whether a frequency is a potential minimizer of $f$. More specifically, since $\ell(\cdot) \leq f(\cdot)$, if for two frequencies $\omega$ and $\omega'$ the inequality $\ell(\omega) > f(\omega')$ holds, then $\omega$ cannot be the frequency that minimizes $f$ since $f(\omega) \geq \ell(\omega) > f(\omega')$. The details for an iterative procedure based on this idea are given in Algorithm \ref{alg:SUBSET}.



\begin{algorithm}[t]
  \caption{Compute the penalized neg. log-likelihood (PNLL) on a subset of $\Omega$}
 \label{alg:SUBSET}
\begin{algorithmic}[1]
  \State compute $\ell(\omega) = \min_{\V{a},\V{\rho},\V{\beta_0}} \ell(\omega,\V{\beta_0},\V{a},\V{\rho})$ at every $\omega \in \Omega$ \Comment{NLL Algorithm} \label{line:i}
\State let $\omega_1,\ldots,\omega_G$ be ordering such that $\ell(\omega_k) \leq \ell(\omega_{k+1})$
\State $i \gets 0$
\State $\Omega_u(0) \gets \varnothing$
\Repeat
\State $i \gets i + 1$
\State $f(\omega_i) \gets \min_{\V{a},\V{\rho},\V{\beta_0}} f(\omega_i,\V{\beta_0},\V{a},\V{\rho})$ \Comment{PNLL BCD Algorithm} \label{line:ii}
\State $\Omega_u(i) = \{\omega \in \Omega: l(\omega) > f(\omega_i)\} \cup \Omega_u(i-1)$ \label{line:prf}
\Until{$\omega_{i+1} \in \Omega_u(i)$ or $i=G$ \label{line:iii}}
\State $R \gets i$
\State $\widehat{\omega} = \argmin{j \in \{1,\ldots,R\}} f(\omega_j)$  \label{line:what}
\end{algorithmic}
\end{algorithm}
The following proposition shows that $\widehat{\omega}$ determined by Line \ref{line:what} of Algorithm \ref{alg:SUBSET} is the global minimizer of $f$ across $\Omega$.
\begin{proposition} Under the notation of Algorithm \ref{alg:SUBSET}
\begin{equation*}
\argmin{j \in \{1,\ldots,R\}} f(\omega_j)  \amp =  \amp \argmin{\omega \in \Omega} f(\omega).
\end{equation*}
\end{proposition}
\begin{proof}
By Line \ref{line:prf} of Algorithm \ref{alg:SUBSET}, for any $r > R$  there exists $h \leq R$ such that $\ell(\omega_r) > f(\omega_h)$. Since $f(\cdot) \geq \ell(\cdot)$ for all $\omega$ we have
\begin{equation*}
f(\omega_r) \amp \geq \amp \ell(\omega_r) \amp > \amp f(\omega_h).
\end{equation*}
Thus $\omega_r$ is not the minimizer of $f$.
\end{proof}

Algorithm \ref{alg:SUBSET} can greatly reduce the number of frequencies for which the PNLL must be minimized. For instance, with the SDSS-II RR Lyrae of Section \ref{sec:sdss} we use an $\Omega$ grid with around 20,000 frequencies for each star. However, by using Algorithm \ref{alg:SUBSET} typically only a few hundred of these frequencies are possible minimizers of the PNLL. Additionally, by running Algorithm \ref{alg:SUBSET} one obtains the $\V{a}$, $\V{\rho}$, and $\V{\beta_0}$ minimizers of the NLL at each $\omega$, which serve as good initialization values in the block coordinate descent algorithm for solving $f(\omega)$, which we turn to next.

\subsection{A Block Coordinate Descent Algorithm}
\label{sec:bcd}
We now describe an algorithm for minimizing the PNLL, \eqref{eq:penlike}, at a particular frequency $\omega$. Since the frequency $\omega$ is fixed in this section, we drop dependence on $\omega$ in the NLL ($\ell(\omega,\V{\beta_0},\V{a},\V{\rho})$ becomes $\ell(\V{\beta_0},\V{a},\V{\rho})$) and PNLL ($f(\omega,\V{\beta_0},\V{a},\V{\rho})$ becomes $f(\V{\beta_0},\V{a},\V{\rho})$). Let $\VE{w}{bi} = \VE{\sigma}{bi}^{-2}$, $s_{bi}(\rho) = \sin(\omega \VE{t}{bi} + \rho)$, and $c_{bi}(\rho) = \cos(\omega \VE{t}{bi} + \rho)$. Thus, we will view the PNLL as the following function of $\V{\beta_0}, \V{a},$ and $\V{\rho}$
\begin{eqnarray}
\label{eq:multiband_problem}
f(\V{\beta_0},\V{a},\V{\rho}) & = & \ell(\V{\beta_0},\V{a},\V{\rho})  + \gamma_1 J(\V{a}) + \gamma_2 J(\V{\rho}),
\label{eq:pen_like}
\end{eqnarray}
where
\begin{eqnarray*}
\ell(\V{\beta_0},\V{a},\V{\rho}) & = & \frac{1}{2}\sum_{b=1}^B \sum_{i=1}^{\VE{n}{b}} \VE{w}{bi}\left [\VE{\beta}{b0} + \VE{a}{b}s_{bi}(\VE{\rho}{b}) - \VE{m}{bi} \right ]^2 \\
J_1(\V{a}) & = & \frac{1}{2} \V{a}\Tra \left (\M{I} - \Vtilde{a}{\Vtilde{a}}\Tra \right ) \V{a} \\
J_2(\V{\rho}) & = & \frac{1}{2} \V{\rho}\Tra \left (\M{I} - \frac{1}{B}\V{1}\V{1}\Tra \right )\V{\rho}.
\end{eqnarray*}
The criterion to be minimized in \eqref{eq:multiband_problem} is non-convex. Consequently, finding a global optimizer may be too much to ask. Nonetheless, we can obtain a local minimizer efficiently using an inexact block coordinate descent (BCD) algorithm. In practice, a local minimizer typically suffices. Our numerical and real data examples demonstrate that the local optimizer provides better solutions than the current state of the art.

We first describe the BCD algorithm at a high level before detailing the individual steps. We minimize the PNLL \eqref{eq:multiband_problem} in round robin fashion with respect to the three blocks of variables $\V{\beta_0}, \V{a},$ and $\V{\rho}$ holding two of them fixed while minimizing with respect to the third one. At the $k+1$th round of updates, we solve in sequence the three smaller optimization problems:

\begin{eqnarray*}
\text{Update 1:} \quad\quad \Vn{\beta}{k+1}_0 & = & \argmin{\V{\beta_0}} \ell(\V{\beta_0},\Vn{a}{k},\Vn{\rho}{k}) \\
\text{Update 2:} \quad\quad \Vn{a}{k+1} & = & \argmin{\V{a}} \ell(\Vn{\beta}{k+1}_0, \V{a}, \Vn{\rho}{k}) + \lambda_1 J_1(\V{a}) \\
\text{Update 3:} \quad\quad \Vn{\rho}{k+1} & = & \argmin{\V{\rho}} \ell(\Vn{\beta}{k+1}_0, \Vn{a}{k+1}, \V{\rho}) + \lambda_2 J_2(\V{\rho}).
\end{eqnarray*}
We clarify that we do not impose constraints on $\V{a}$ and $\V{\rho}$ in performing individual block updates, but rather impose the constraints at the termination of the BCD algorithm. We discuss later why we can make this simplification, when we derive the individual block updates.

Our motivation for using BCD is that the individual updates are simple and fast to compute. We next derive these updates. To streamline our derivations we define a few terms more compactly.
Let $\V{s}_b(\rho) = (s_{1b}(\rho), \ldots, s_{n_bb}(\rho))\Tra$. Define $\V{c}_b(\rho)$ analogously. Let $\V{m}_b = (\VE{m}{1b}, \ldots, \VE{m}{n_bb})\Tra$.


\subsubsection*{Update 1: the intercepts $\V{\beta_0}$}

Fix $\V{a}$ and $\V{\rho}$. Let $\V{\beta_0}^+ = \arg\min_{\V{\beta_0}} f(\V{\beta_0},\V{a},\V{\rho})$ denote the update for $\V{\beta_0}$. The update for $\V{\beta_0}$ separates in its $B$ elements and is given by the minimizers to $B$ univariate convex quadratic functions
\begin{eqnarray}
\label{eq:update1}
\VE{\beta}{b0}^+ & = & \argmin{\beta} [\beta\V{1} + \VE{a}{b}\V{s}_b(\VE{\rho}{b}) - \V{m}_b]
\Tra \M{W}_b [\beta\V{1} + \VE{a}{b}\V{s}_b(\VE{\rho}{b}) - \V{m}_b],
\end{eqnarray}
where $\M{W}_b \in \Real^{\VE{n}{b} \times \VE{n}{b}}$ is a diagonal matrix with $i$th diagonal element $\VE{w}{bi}$.
It is straightforward to show that the solution to the optimization problem \eqref{eq:update1} is given by
\begin{eqnarray}
\label{eq:update_beta0}
\VE{\beta}{b0}^+ & = & \frac{\langle \V{1}, \M{W}_b[\V{m}_b - \VE{a}{b}\V{s}_b(\VE{\rho}{b})] \rangle}{\langle \V{1}, \M{W}_b\V{1} \rangle}.
\end{eqnarray}

\subsubsection*{Update 2: the amplitudes $\V{A}$}
 
Fix $\V{\beta_0}$ and $\V{\rho}$.  Let $\V{A}^+ = \arg\min_{\V{A}} f(\V{\beta_0},\V{A},\V{\rho})$ denote the update for $\V{A}$. To obtain the update $\V{a}^+$, we minimize a convex quadratic function
\begin{eqnarray}
\label{eq:update2}
\V{a}^+ & = &  \argmin{\V{a}} \frac{1}{2}\sum_{b=1}^B \sum_{i=1}^{n_b} \VE{w}{bi} \left [\VE{a}{b}s_{bi}(\VE{\rho}{b}) - \VE{\mu}{bi} \right ]^2 +  \frac{\gamma_1}{2} 
\V{a}\Tra[\M{I} - \Vtilde{a}\Vtilde{a}\Tra]\V{a},
\end{eqnarray}
where $\VE{\mu}{bi} = \VE{m}{bi} - \VE{\beta}{b0}$. Let $\V{\mu}_b = \V{m}_b - \VE{\beta}{b0}\V{1}$.
Setting the derivative of the quadratic function in \eqref{eq:update2} with respect to $\VE{a}{b}$ equal to zero yields the $B$ stationary conditions
\begin{eqnarray*}
\V{s}_{b}(\VE{\rho}{b})\Tra\M{W}_b\V{s}_{b}(\VE{\rho}{b}) \VE{a}{b} - \langle \V{s}_{b}(\VE{\rho}{b}), \M{W}_b\V{\mu}_b \rangle + \gamma_1\left [ \VE{a}{b} - \langle \Vtilde{a}, \V{a} \rangle \tilde{a}_{b} \right] & = & 0.
\end{eqnarray*}
The stationarity conditions imply that the update $\V{A}^+$ is the solution to the following linear system of equations
\begin{eqnarray*}
\left [\M{E} - \gamma_1\Vtilde{a}\Vtilde{a}\Tra \right] \V{A}
& = & \V{\xi},
\end{eqnarray*}
where $\M{E}$ is a diagonal $B \times B$ matrix with $b$th diagonal element
$\ME{E}{bb} = \V{s}_b(\VE{\rho}{b})\Tra\M{W}_b\V{s}_b(\VE{\rho}{b}) + \gamma_1$ and
\begin{eqnarray*}
\V{\xi} & = &
\begin{pmatrix}
\langle \V{s}_1(\rho_1), \M{W}_1\V{\mu}_1 \rangle \\
\vdots \\
\langle \V{s}_B(\VE{\rho}{B}), \M{W}_B\V{\mu}_B \rangle \\
\end{pmatrix}.
\end{eqnarray*}
Applying the matrix inversion lemma enables us to solve the system with a single matrix-vector multiply
\begin{eqnarray}
\label{eq:update_a}
\V{A}^+ & = & \M{E}\Inv \left [\M{I} + \frac{1}{\frac{1}{\gamma_1} - \Vtilde{a}\Tra\M{E}\Inv\Vtilde{a}} \Vtilde{a}\Vtilde{a}\Tra\M{E}\Inv \right ]
\V{\xi}.
\end{eqnarray}

We now address why we are able to drop the nonnegativity constraints on $\V{a}$. Suppose that the mean brightness for a star at time $t$ in some band varies according to
\begin{eqnarray*}
\mu(t) & = & a \sin(\omega t + \rho),
\end{eqnarray*}
where the amplitude $a$ is negative. Then
\begin{eqnarray*}
\mu(t) & = & -\lvert a \rvert \sin(\omega t + \rho) \amp = \amp \lvert a \rvert \sin(\omega t + \rho + \pi)
\end{eqnarray*}
and so having a ``negative amplitude'' $a$ with phase $\rho$ is equivalent to having a positive amplitude $\lvert a \rvert$ with phase $\rho + \pi$. Consequently, if at the termination of PGLS an amplitude is negative, we can simply shift the estimate of the corresponding phase by $\pi$. 

\subsubsection*{Update 3: the phases $\V{\rho}$ via MM}

Fix $\V{\beta_0}$ and $\V{A}$. Unlike the previous two updates, which required minimizing a convex function, the update for $\V{\rho}$ requires minimizing a non-convex function.
Nonetheless, it can be effectively attacked using a majorization-minimization (MM) algorithm \citep{BecYanLan1997,LanHunYan2000} to get an approximate solution to the problem
\begin{eqnarray}
\label{eq:update_rho}
\Vn{\rho}{k+1}_0 & = & \argmin{\V{\rho}} \ell(\Vn{\beta}{k+1}_0, \Vn{a}{k+1}, \V{\rho}) + \lambda_2 J_2(\V{\rho}).
\end{eqnarray}

The basic principle behind an MM algorithm is to convert a hard optimization problem into a sequence of simpler ones. The MM principle requires majorizing the objective function $f(\V{y})$ by a surrogate function $g(\V{y} \mid \V{x})$ anchored at the current point $\V{x}$.  Majorization is a combination of a tangency condition $g(\V{x} \mid \V{x}) =  f(\V{x})$ and a domination condition $g(\V{y} \mid \V{x})  \geq f(\V{y})$ for all $\V{y} \in \Real^d$.  The associated MM algorithm is then defined by the iterates
\begin{eqnarray}
  \label{eq:generic-MM-iterate}
  \V{x}_{k+1} & := & \argmin{\V{y}} g(\V{y} \mid \V{x}_{k}).
\end{eqnarray}
Because 
\begin{equation}
  f(\V{x}_{k+1}) \amp \leq \amp g(\V{x}_{k+1} \mid \V{x}_{k}) \amp \leq \amp g(\V{x}_{k} \mid \V{x}_{k}) \amp = \amp f(\V{x}_{k}),
\end{equation}
the MM iterates generate a descent algorithm driving the objective function downhill.

To update $\V{\rho}$, we take advantage of the fact that for each $b$
\begin{eqnarray*}
f_{b}(\VE{\rho}{b}) & = & \frac{1}{2}\sum_{i=1}^{n_b} \VE{w}{bi}\left [\VE{a}{b}s_{bi}(\VE{\rho}{b}) - \VE{\mu}{bi} \right ]^2
\end{eqnarray*}
is Lipschitz differentiable. This fact can be leveraged to construct a simple convex quadratic majorization of $f$ as a function of $\V{\rho}$.
\begin{proposition}
\label{prop:majorization}
The following function majorizes $\sum_b f_{b}(\VE{\rho}{b})$ at the point $\Vtilde{\rho}$:
\begin{eqnarray*}
g(\V{\rho} \mid \Vtilde{\rho})
& = & \sum_{b=1}^B \left [
f_b(\tilde{\rho}_b) + f'(\tilde{\rho}_b) (\VE{\rho}{b} - \tilde{\rho}_b) + \frac{L_b}{2} (\VE{\rho}{b} - \tilde{\rho}_b)^2
\right ],
\end{eqnarray*}
where
\begin{eqnarray*}
L_b & = & \VE{a}{b}  \left [\VE{a}{b} \kappa_b + \sqrt{n_b}\lVert \M{W}_b\V{\mu}_{b} \rVert_2 \right ] \quad \text{and} \quad \kappa_b \amp = \amp \V{1}\Tra \M{W}_b \V{1}.
\end{eqnarray*}
\end{proposition}
The proof is given in the \App{mm}. The MM algorithm generates an improved estimate $\V{\rho}^+$ of the solution to \eqref{eq:update_rho} given a previous estimate
$\Vtilde{\rho}$ using the following update rule
\begin{eqnarray}
\label{eq:mm_update}
\V{\rho}^+ & = & \argmin{\V{\rho}} g(\V{\rho} \mid \Vtilde{\rho}) + \gamma_2 J_2(\V{\rho}).
\end{eqnarray}
There is always a unique minimizer $\V{\rho}^+$ to \eqref{eq:mm_update}, since $g(\V{\rho} \mid \Vtilde{\rho})$ is strongly convex in $\V{\rho}$.

We now show how $\V{\rho}^+$ can be expressed explicitly in terms of $\Vtilde{\rho}$. Note that
\begin{eqnarray}
\label{eq:maj_stationarity}
\frac{\partial}{\partial \VE{\rho}{b}} \left [
g(\V{\rho} \mid \Vtilde{\rho}) + \gamma_2 J_2(\V{\rho})
\right ] 
& = & f'_b(\tilde{\rho}_b) + L_b(\VE{\rho}{b} - \tilde{\rho}_b) + \gamma_2 \left [ \VE{\rho}{b} - \frac{1}{B}\sum_{b' =1}^B \rho_{b'} \right ].
\end{eqnarray}
The stationarity conditions imply that the update $\V{\rho}^+$ is the solution to the linear system of equations
\begin{eqnarray}
\label{eq:mm_system}
\left [\M{F} - \frac{\gamma_2}{B} \V{1}\V{1}\Tra \right] \V{\rho}
& = & \V{\zeta}(\Vtilde{\rho}),
\end{eqnarray}
where $\M{F}$ is a diagonal $B \times B$ matrix with $b$th diagonal element
$\ME{F}{bb} = L_b + \gamma_2$ and
\begin{eqnarray*}
\V{\zeta}(\Vtilde{\rho}) & = &
\begin{pmatrix}
L_1\tilde{\rho}_1 - f'_1(\tilde{\rho}_1) \\
\vdots \\
L_B\tilde{\rho}_B - f'_B(\tilde{\rho}_B) \\
\end{pmatrix}.
\end{eqnarray*}
Again we turn to the matrix inversion lemma to solve \eqref{eq:mm_system} with a single matrix-vector multiply
\begin{eqnarray}
\label{eq:update_rho_explicit}
\V{\rho}^+ & = & \M{F}\Inv \left [\M{I} + \frac{1}{\frac{B}{\gamma_2} - \V{1}\Tra\M{F}\Inv\V{1}} \V{1}\V{1}\Tra\M{F}\Inv \right ]
\V{\zeta}(\Vtilde{\rho})
\end{eqnarray}

Similar to the nonnegativity constraint on the amplitude, we ``enforce" the box constraint on the phase at the termination of the BCD algorithm. Since the criterion \eqref{eq:multiband_problem} is periodic in $\V{\rho}$, we choose the solution that is within the constraint set $[0,\pi]^B$.

Finally, we note that it is sufficient to perform a single MM update \eqref{eq:update_rho_explicit} and set $\Vtilde{\rho} = \Vn{\rho}{k}$ and $\Vn{\rho}{k+1} = \V{\rho}^+$. Applying multiple MM updates \eqref{eq:update_rho_explicit} to get a more precise update for $\V{\rho}$ typically does not pay in practice as $\V{\beta_0}$ and $\V{a}$ may change quite a bit during early rounds of updates. Moreover, as discussed below, taking only a single MM update is taken per round of block coordinate updates does not change the overall convergence of the BCD algorithm.

\subsubsection*{Complexity and Convergence of BCD}

\begin{algorithm}[t]
Initialize $\Vn{\beta}{0}, \Vn{A}{0},$ and $\Vn{\rho}{0}$ as the solutions to multiband GLS (MGLS).
\begin{algorithmic}[1]
  \caption{Block Coordinate Descent + MM for penalized GLS (PGLS)}
  \label{alg:BCDMM}
\State $k \gets 0$
\Repeat 
\For{$b = 1, \ldots, B$} \Comment{Update $\V{\beta_0}$} \label{line:a}
\State $\beta^{(k+1)}_{b0} \gets \langle \V{1}, \M{W}_b[\V{m}_b - \VnE{a}{k}{b}\V{s}_b(\VE{\rho}{b}^{(k)})] \rangle/\langle \V{1}, \M{W}_b\V{1} \rangle$
\EndFor
\State Update $\M{E} \in \Real^{B \times B}$ and $\V{\xi} \in \Real^B$ with $\Vn{\rho}{k}$ and $\Vn{\beta}{k+1}_0$ \Comment{Update $\V{a}$} 
\State $\Vn{A}{k+1} \gets \left [\M{E} - \gamma_1 \Vtilde{a}\Vtilde{a}\Tra \right]\Inv \V{\xi}$ \label{line:b}
\State Update $\M{F} \in \Real^{B \times B}$ and $\V{\zeta} \in \Real^B$ with $\Vn{a}{k+1}, \Vn{\beta}{k+1}_0$, and $\Vn{\rho}{k}$ \Comment{Update $\V{\rho}$} \label{line:c}
\State $\Vn{\rho}{k+1} \gets \left [\M{F} - \gamma_2\V{1}\V{1}\Tra \right ] \Inv \V{\zeta}$
\State $k \gets k+1$
\Until{convergence}
\end{algorithmic}
\end{algorithm}

\Alg{BCDMM} provides pseudocode of the BCD algorithm. We update the intercepts $\V{\beta_0}$, amplitudes $\V{a}$, and phases $\V{\rho}$ in round robin fashion until the relative change in the variables falls below a defined tolerance.

A single round of block updates is computationally efficient. Updating $\V{\beta_0}$ requires $\mathcal{O}(N)$ operations where $N = \sum_{b=1}^B \VE{n}{b}$. Updating $\V{a}$ requires computing the diagonal matrix $\M{E} \in \Real^{B \times B}$ and vector $\V{\xi} \in \Real^N$ which requires $\mathcal{O}(N)$ operations. Solving the linear system in Line~7 of \Alg{BCDMM} requires $\mathcal{O}(B)$ operations according to the explicit update given in \Eqn{update_a}. Similarly, setting up and solving the linear system in Line~9 to update $\V{\rho}$ requires $\mathcal{O}(N)$ and $\mathcal{O}(B)$ operations. Thus, the total amount of work per round of block coordinate descent is $\mathcal{O}(N)$ or in other words linear in the size of the data.


We end this section with the following convergence result for \Alg{BCDMM} 
\begin{proposition}
\label{prop:convergence}
The iterates $(\Vn{\beta_0}{k},\Vn{a}{k},\Vn{\rho}{k})$ of \Alg{BCDMM} tend to stationary points of the PNNL at a fixed frequency $\omega$.
\end{proposition}
The proof is contained in \App{convergence}.

\section{Selection of $\gamma_1$ and $\gamma_2$}
\label{sec:gamma_selection}

\begin{figure}[t]
  \begin{center}
    \begin{includegraphics}[scale=0.35]{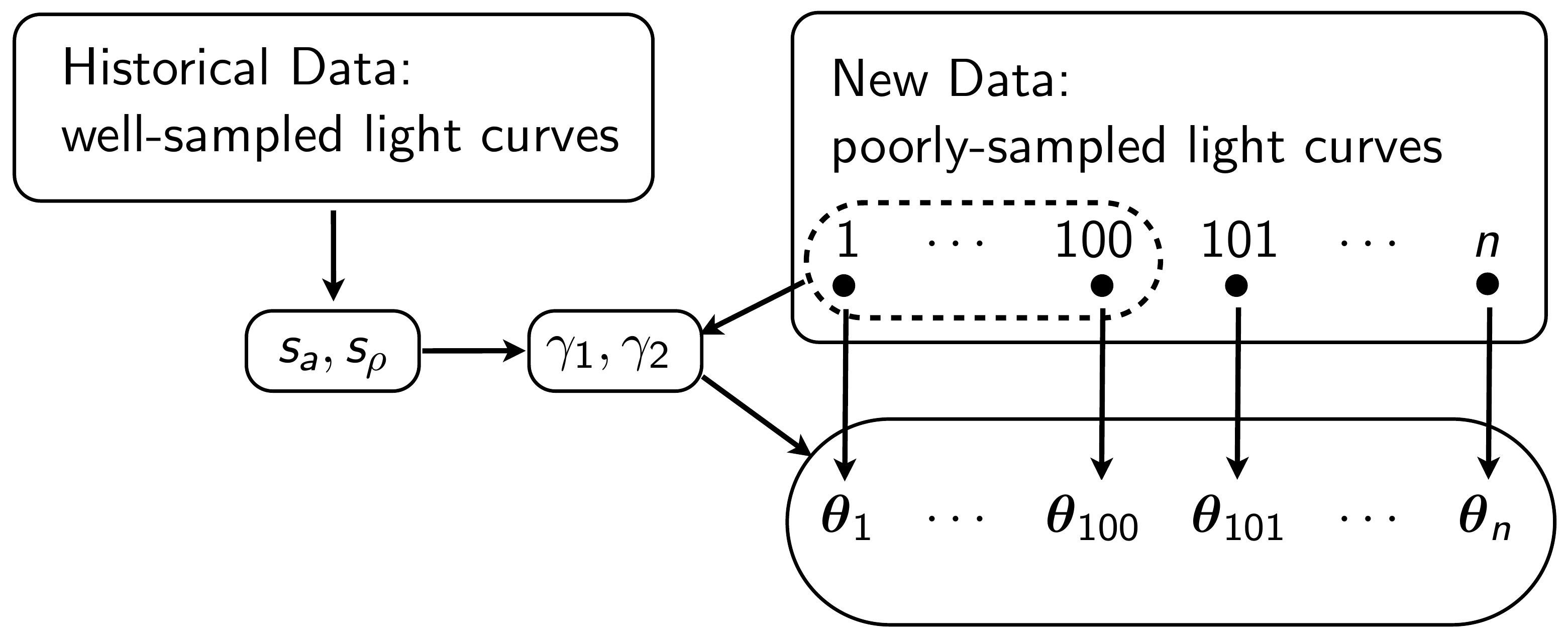}
      \caption{Historical periodic variable star data comprised of well--sampled light curves provide a guide towards selecting tuning parameters. The scatter variables $s_a$ and $s_\rho$ are estimated on historical data. A subset of the new data of poorly--sampled light curves, for example the first 100 light curves, is used to select the tuning parameters $\gamma_1$ and $\gamma_2$. We set $\gamma_1$ and $\gamma_2$ such that the scatter of amplitudes and phases approximately matches the scatter of the historical data. The $i$th lightcurve in the new data is used to estimate the $i$th set of model parameters $\V{\theta}_i$, namely the $i$th period, amplitude, phase, and intercept.\label{fig:selection}}
    \end{includegraphics}
  \end{center}
\end{figure}

Judiciously choosing the amount of regularization in PGLS is critical to its performance. In practice, regularization parameters are often selected using cross-validation \citep{HasTib2001}. The main drawback to cross-validation is that it is computationally intensive. Here we propose a computationally efficient and simple data-driven alternative that leverages the availability of historical data of well--observed periodic variables.

Recall that the regularization parameter $\gamma_1$ controls the degree of shrinkage in the amplitude vector $\V{a}$. Small values of $\gamma_1$ lead to less shrinkage (amplitudes far from a multiple of $\Vtilde{a}$), and large values of $\gamma_1$ lead to more shrinkage (amplitudes close to a multiple of $\Vtilde{a}$). In astronomy, there is an abundance of well--observed light curves where methods such as MGLS correctly estimate periods, amplitudes, and phases. Using this historical data we can estimate the scatter of the amplitudes around $\Vtilde{a}$ and the phases around $\V{1}$. We tune $\gamma_1$ and $\gamma_2$ so that the scatter for the poorly observed amplitude and phase estimates approximately matches the scatter observed in the historical data.

We detail this process now for selecting $\gamma_1$. Selection of $\gamma_2$ is analogously accomplished by substituting the vector $B^{-1/2}\V{1}$ for the vector $\Vtilde{a}$. Let $\V{a}'_1,\ldots,\V{a}'_m$ be the amplitudes from a historical set of well--observed light curves. One could use MGLS to estimate these quantities. We assume these quantities have little measurement error, because the light curves have been well--observed. In the simulated and real data examples of Section \ref{sec:data}, we assume we have access to $m=100$ light curves. Let $\Vtilde{a}$ be the normalized mean of the amplitudes and define the scatter of the amplitudes around $\Vtilde{a}$ as
\begin{equation}
\label{eq:s_a}
s_a \amp = \amp \median{i \in \{1,\ldots,m\}}\V{a'}_i\Tra \left (\M{I} - \Vtilde{a}{\Vtilde{a}}\Tra \right ) \V{a}_i'.
\end{equation}

Given a set of poorly observed light curves $D_1,\ldots,D_n$, we expect the scatter in the amplitudes to approximately match $s_a$. For some trial value of the tuning parameter $\gamma_1$ define the amplitude fit for $D_i$ using PGLS as
\begin{align*}
\V{a}_i(\gamma_1) &\amp= \amp \argmin{\V{a}} \min_{\omega,\V{\beta_0},\V{\rho}} f(\omega,\V{\beta_0},\V{a},\V{\rho} \mid D_i;\gamma_1,0)
\end{align*}
and the resulting scatter in amplitudes as
\begin{equation}
\label{eq:sa}
s_a(\gamma_1) \amp = \amp \median{i\in\{1,\ldots,n\}} \V{a}_i(\gamma_1)\Tra \left (\M{I} - \Vtilde{a}{\Vtilde{a}}\Tra \right ) \V{a}_i(\gamma_1).
\end{equation}
As $\gamma_1$ increases, the amplitude estimates $\V{a}_i(\gamma_1)$ are pulled towards $\Vtilde{a}$, causing $s_a(\gamma_1)$ to decrease. As $\gamma_1$ decreases, the amplitude estimates $\V{a}_i(\gamma_1)$ are pushed away from $\Vtilde{a}$, causing $s_a(\gamma_1)$ to increase. We tune $\gamma_1$ such that $s_a(\gamma_1)$ is approximately equal to $s_a$. Since $s_a(\gamma_1)$ is inversely proportional to $\gamma_1$, we perform a binary search to find the optimal value of $\gamma_1$, using a $\log$ linear grid to find initial lower and upper bounds. We terminate the search when an update to $\gamma_1$ does not alter any of the period estimates more than 1\%, implying that further changes to $\gamma_1$ will not significantly alter the resulting period estimates. An analogous procedure is used for selecting $\gamma_2$. Since the total number of stars for which we want to estimate periods can be large, we select 100 light curves for determining $s_a(\gamma_1)$ in \eqref{eq:sa}, rather than the full set of light curves for which we seek to estimate periods. A schematic of how historical and new data are used to estimate the various parameters is summarized in Figure~\ref{fig:selection}.

\section{Data Analysis}
\label{sec:data}
\subsection{Simulations}

\begin{figure}[t]
\centering
 \subfloat[Simulated phases.]{\label{fig:simulated_phases}
\includegraphics[scale=.42]{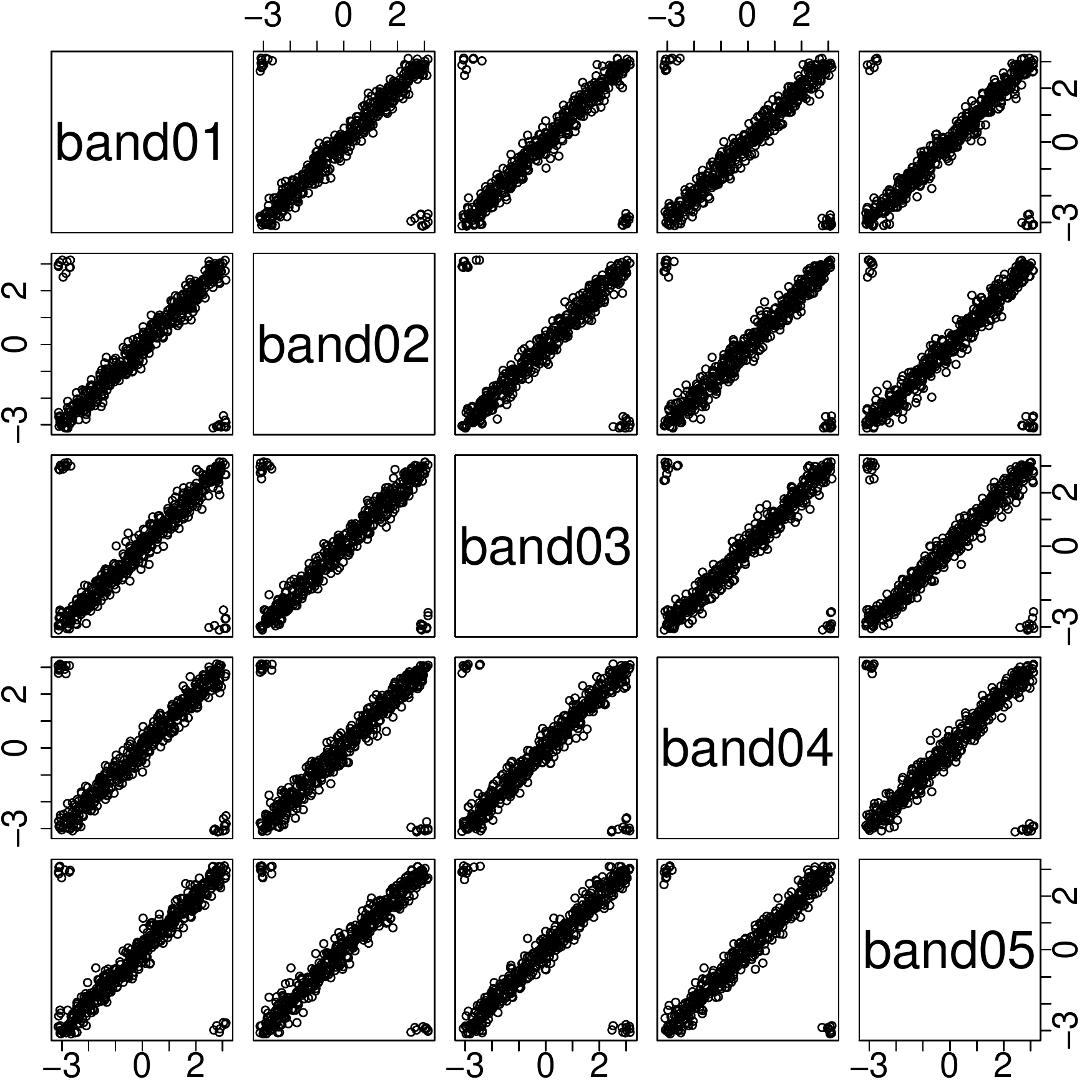}
 } 
 \subfloat[Simulated amplitudes.]{\label{fig:simulated_amps}
\includegraphics[scale=.42]{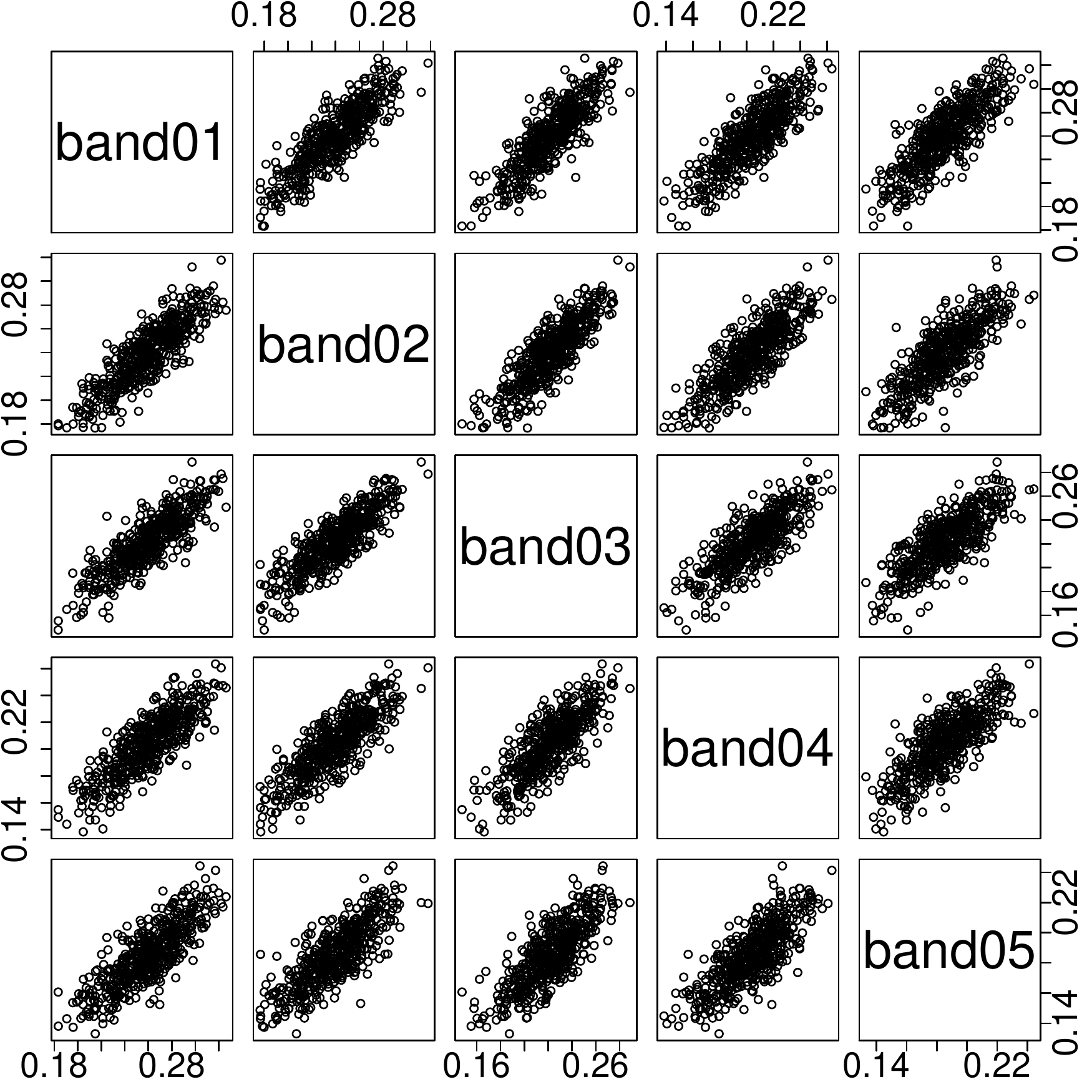}
}
 \caption{Correlations in simulated (a) phases and (b) amplitudes. These approximately match correlations observed in real data, such as in Figures \ref{fig:scatter_A} and \ref{fig:scatter_rho}.}
\end{figure}

We compare MGLS and PGLS by simulating sinusoids and determining period estimation accuracy of each method. We generate 500 5--band light curves from the sinusoidal likelihood model. Figures \ref{fig:simulated_phases} and \ref{fig:simulated_amps} show the distribution of phases and amplitudes. These amplitudes and phases are drawn from distributions meant to approximately match correlations seen in real data such as Figures \ref{fig:scatter_A} and \ref{fig:scatter_rho}. The periods are drawn uniform on $[0.2,1.0]$ days. This represents the plausible range of periods for RR Lyrae type stars which we study using real data in the following section.
  \begin{figure}[p]
\centering
 \subfloat[MGLS with 5 points.]{\label{fig:simulated_pvsp5gls}
\includegraphics[scale=.4]{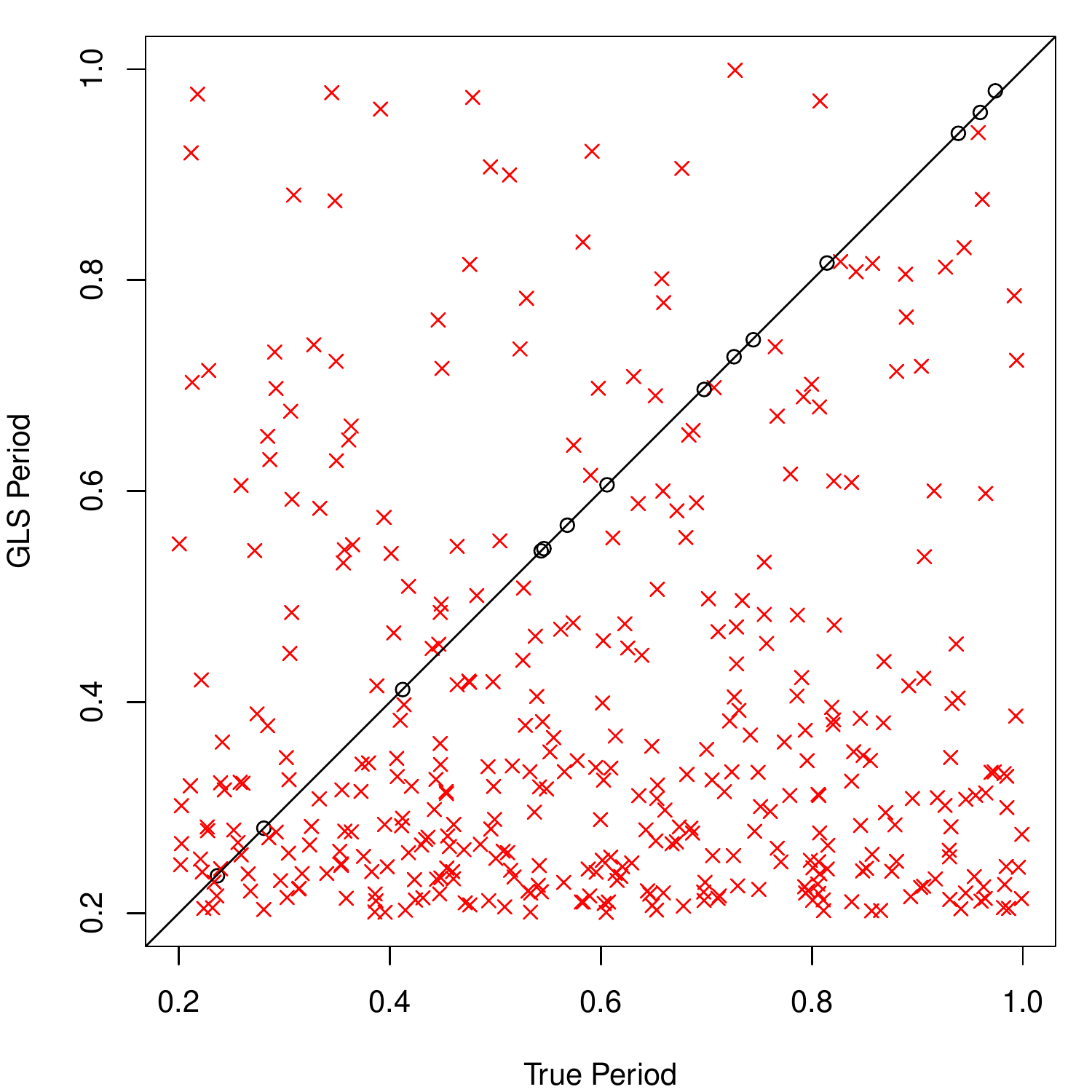}
 } 
 \subfloat[PGLS with 5 points.]{\label{fig:simulated_pvsp5pgls}
\includegraphics[scale=.4]{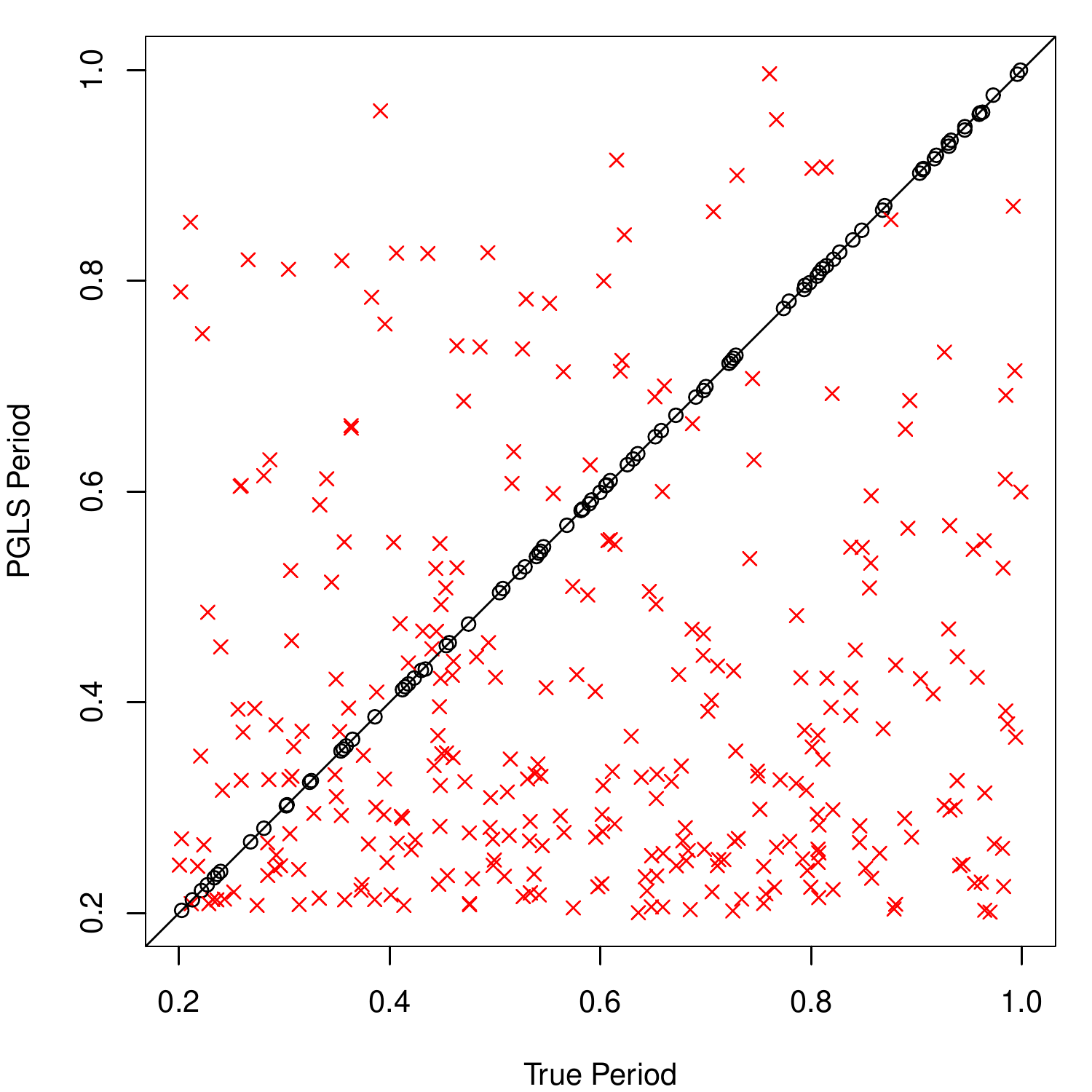}
}\\
 \subfloat[MGLS with 10 points.]{\label{fig:simulated_pvsp10gls}
\includegraphics[scale=.4]{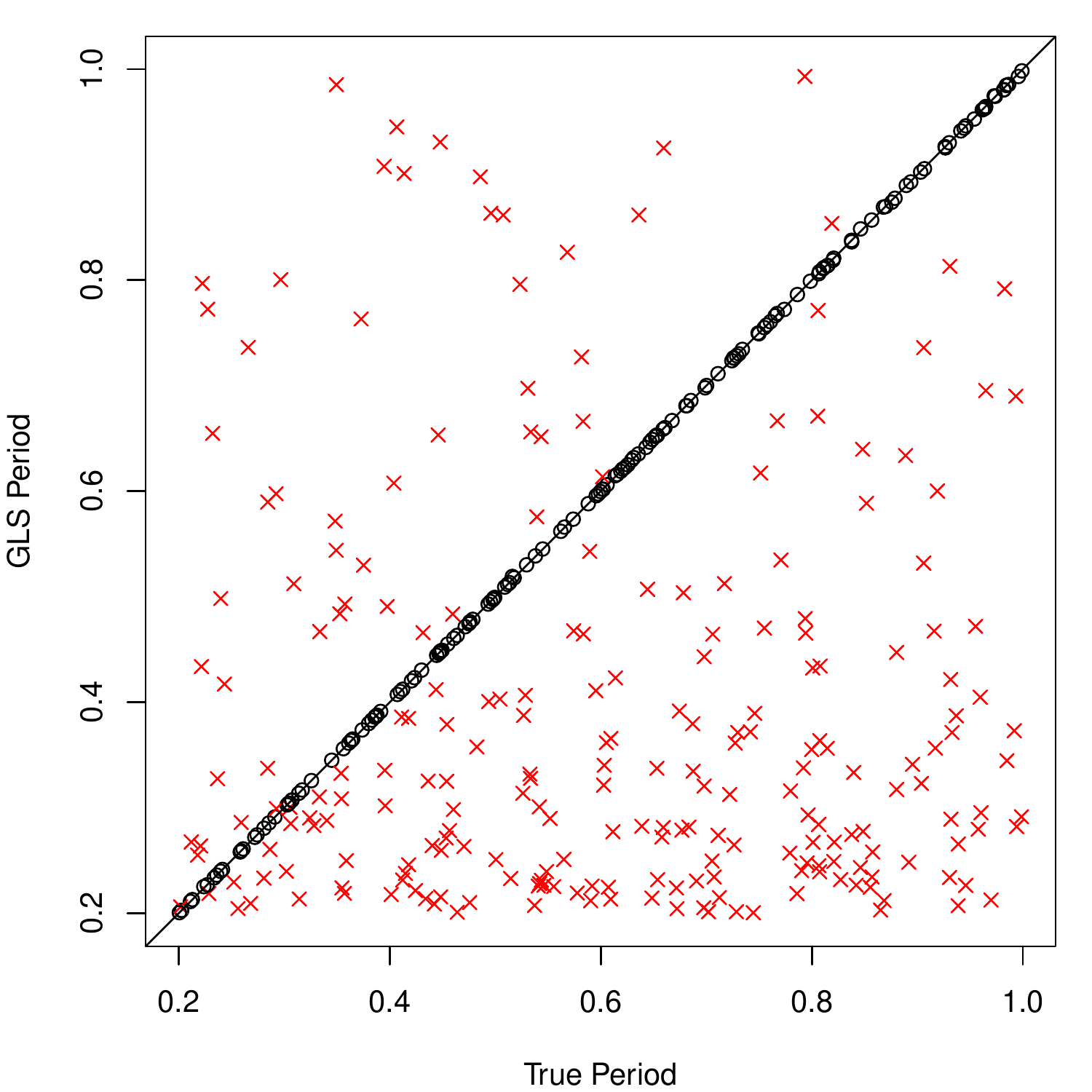}
 } 
 \subfloat[PGLS with 10 points.]{\label{fig:simulated_pvsp10pgls}
\includegraphics[scale=.4]{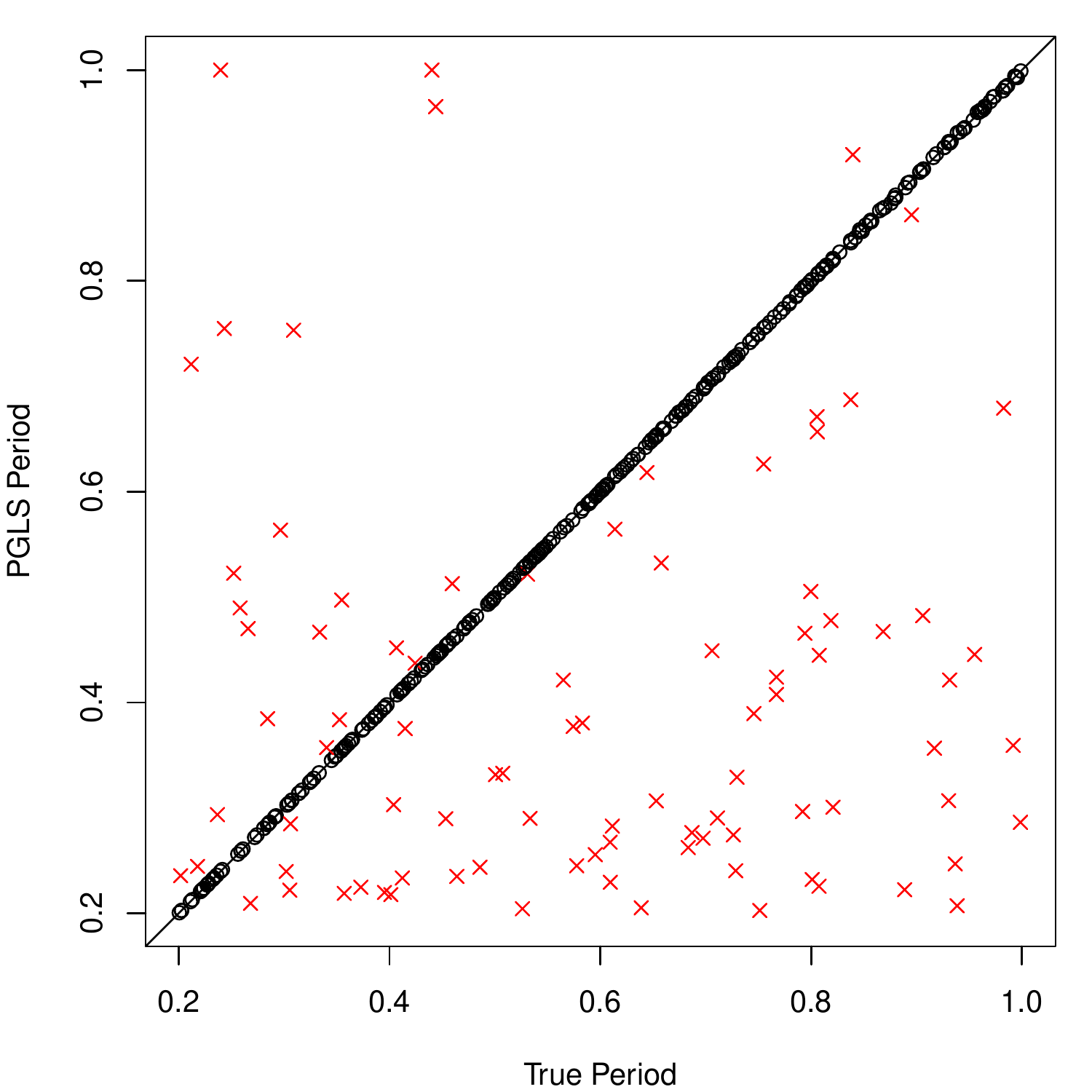}
}\\
 \subfloat[MGLS with 15 points.]{\label{fig:simulated_pvsp15gls}
\includegraphics[scale=.4]{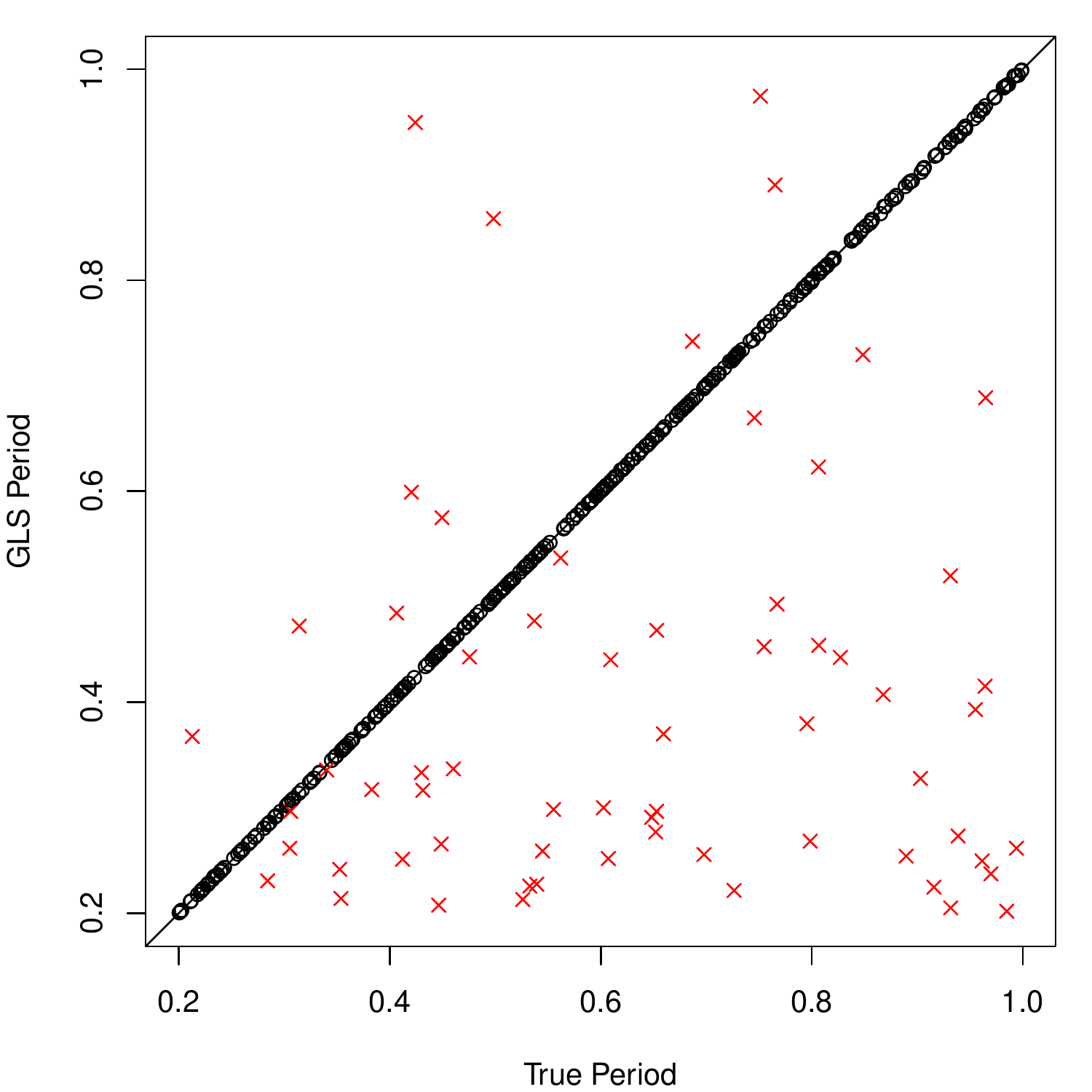}
 } 
 \subfloat[PGLS with 15 points.]{\label{fig:simulated_pvsp15pgls}
\includegraphics[scale=.4]{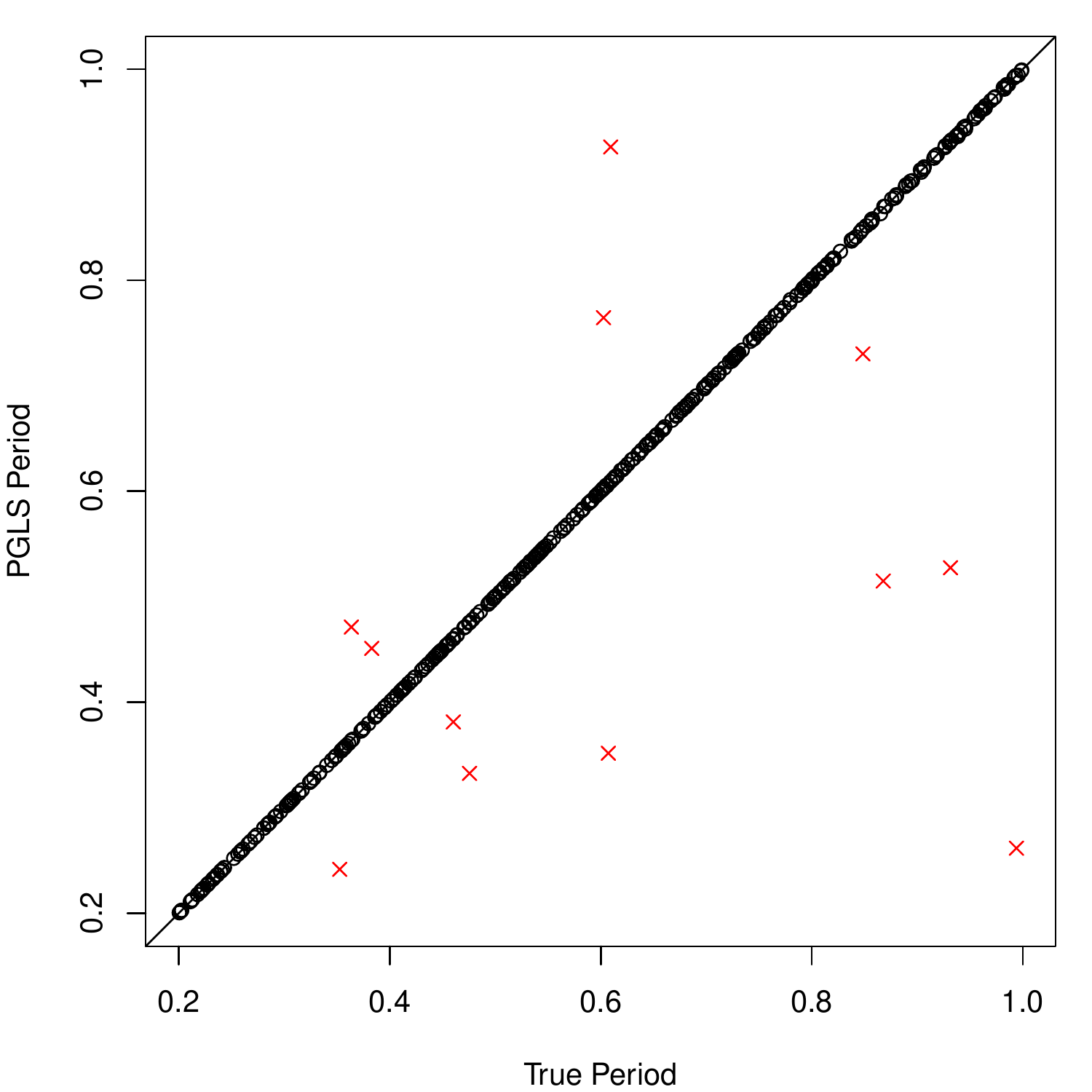}
}
 \caption{Scatterplots of true period versus estimate for MGLS and PGLS with 5, 10, and 15 observations per band for the test set. The MGLS estimates are in the left column and the PGLS estimates are in the right column. The red $\times$ are estimates further than 1\% away from truth while the black $\circ$ are estimates within 1\% of truth. PGLS estimates are significantly more accurate.\label{fig:accuracies}}
\end{figure}

We divide the light curves into two groups: 100 historical/training light curves used for estimating $\tilde{\V{a}}$ and $s_a$ (see \eqref{eq:s_a}) and 400 light curves for testing. We downsample the test light curves to 5, 10, and 15 observations per curve to simulate difficult period recovery regimes. For each number of observations per light curve we tune $\gamma_1$ and $\gamma_2$ according to Section \ref{sec:gamma_selection}.

Scatterplots of true period versus period estimate for MGLS and PGLS for 5, 10, and 15 observations per band are contained in Figure \ref{fig:accuracies}. PGLS improves period estimates in each case.

In Figure \ref{fig:sim_gls_pgls_rho_correlations} we plot phase estimates for MGLS and PGLS for 5 and 10 observation per band test sets. As expected PGLS phases are more concentrated around the vector $\V{1}$. For 15 observations per band, MGLS estimates show significant concentration around $\V{1}$, indicative that MGLS is estimating phases correctly with 15 observations per band. Notice that the PGLS eliminates phase estimates near $(-\pi,\pi)$ and $(\pi,-\pi)$. Since these are plausible phases, this tendancy could result in incorrect period estimates for some stars. A more refined penalty term for phase ($J_2$) in PGLS could address this issue. 
\begin{figure}[t]
\centering
 \subfloat[5 observations.]{\label{fig:simulated_phases_5}
\includegraphics[scale=.42]{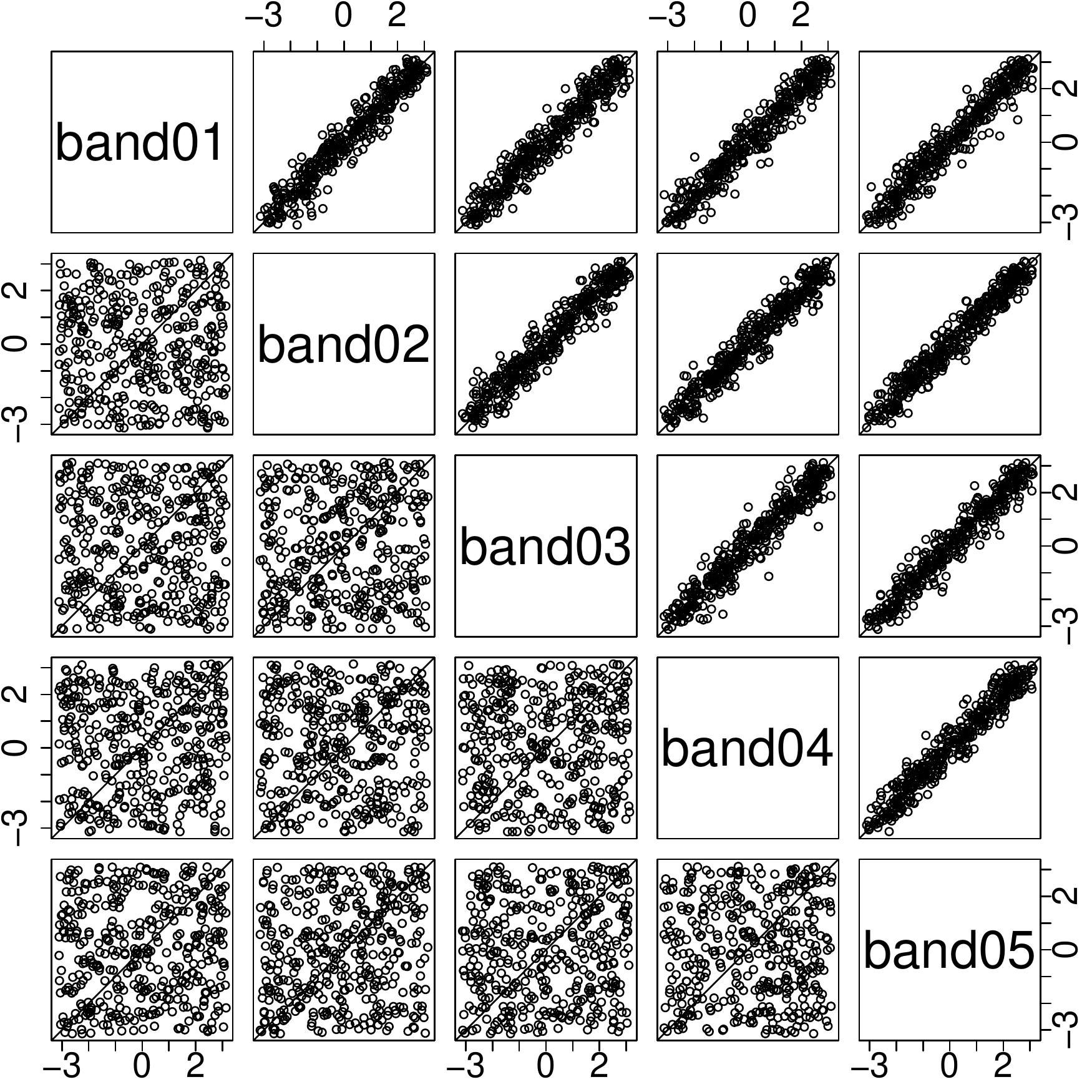}
 } 
 \subfloat[15 observations.]{\label{fig:simulated_phases_15}
\includegraphics[scale=.42]{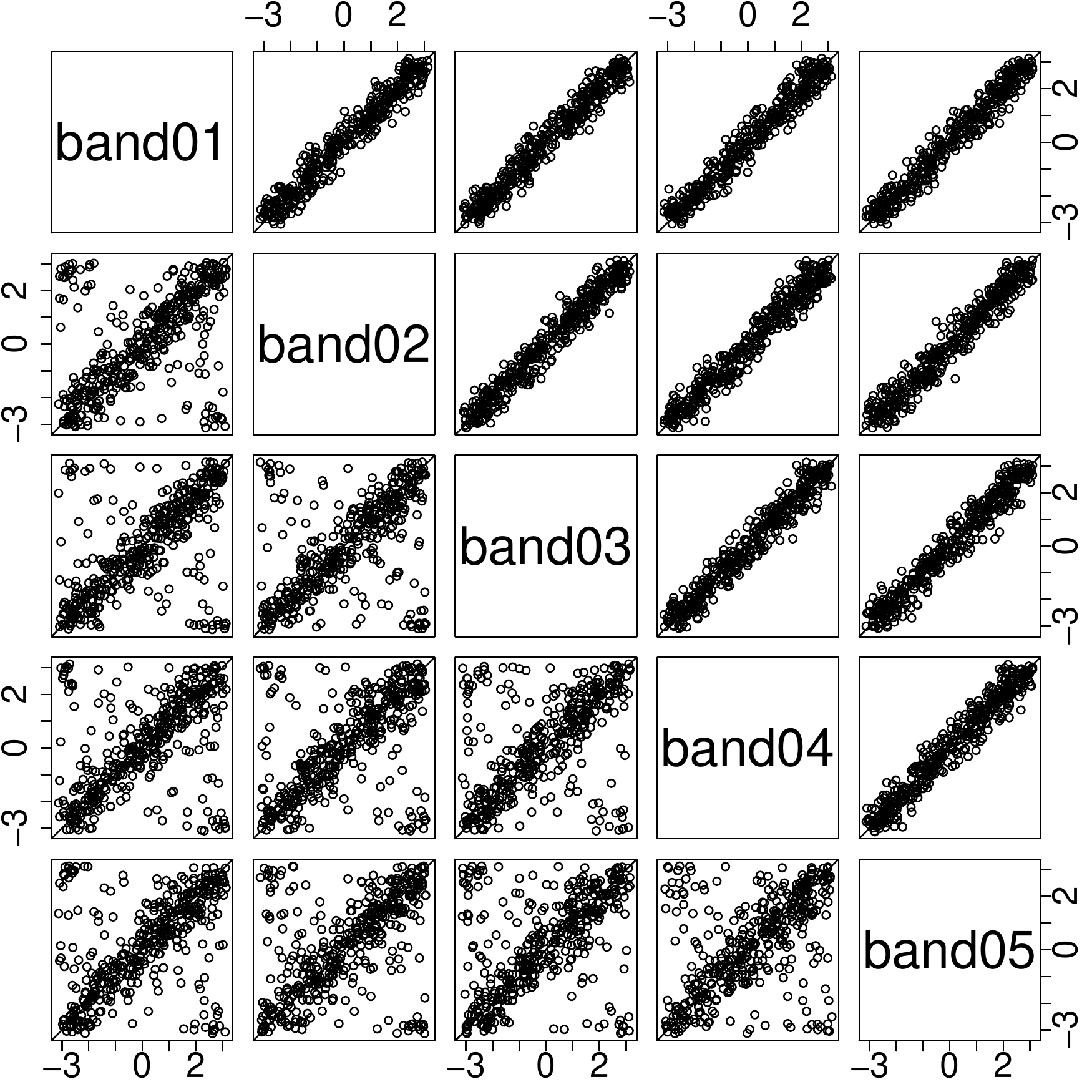}
}
 \caption{Correlations in phase estimates with (a) 5 observations per band and (b) 15 observations per band. The upper diagonal in each plot are PGLS estimates. The lower diagonal are MGLS estimates. The phases estimates using PGLS appear more realistic than MGLS, especially for light curves with 5 observations per band.\label{fig:sim_gls_pgls_rho_correlations}}
\end{figure}



\subsection{SDSS Stripe--82 Data}
\label{sec:sdss}

\cite{sesar2010light} identified 483 periodic variable stars of the class RR Lyrae in the Sloan Digital Sky Survey II (SDSS-II). These light curves were sampled approximately 30 times per band in five bands. \cite{sesar2010light} estimated periods for these stars using a Supersmoother routine in \cite{reimann1994frequency}. Visual inspection suggests these period estimates are accurate.

While Supersmoother is accurate with well--sampled light curves, with poorly--sampled light curves it suffers the same problems as MGLS. This is why \cite{sesar2007exploring} did not attempt period estimation when SDSS-I collected a median of only 10 observations per band for these stars. Ongoing surveys, such as PanSTARRS1, currently have data roughly the quality of SDSS-I in terms of number of observations per band \citep{schlafly2012photometric}. Here we study period estimation with poorly--sampled light curves by downsampling SDSS-II data and comparing MGLS and PGLS period estimates to the Supersmoother estimates computed using the entire light curves.

\input{figs/sdss_results_output_results}

We obtained 450 of 483 \cite{sesar2010light} RR Lyrae light curves from a public data repository \citep{ivezic2007sloan}. To test period estimation with poorly--sampled light curves we split the 450 light curves into 100 training and 350 test. As with the simulated data, we downsample the test light curves to 5, 10, and 15 measurements per band and compare MGLS and PGLS for period recovery.

\begin{figure}[p]
\centering
 \subfloat[MGLS with 5 points.]{\label{fig:sdss_pvsp5gls}
\includegraphics[scale=.4]{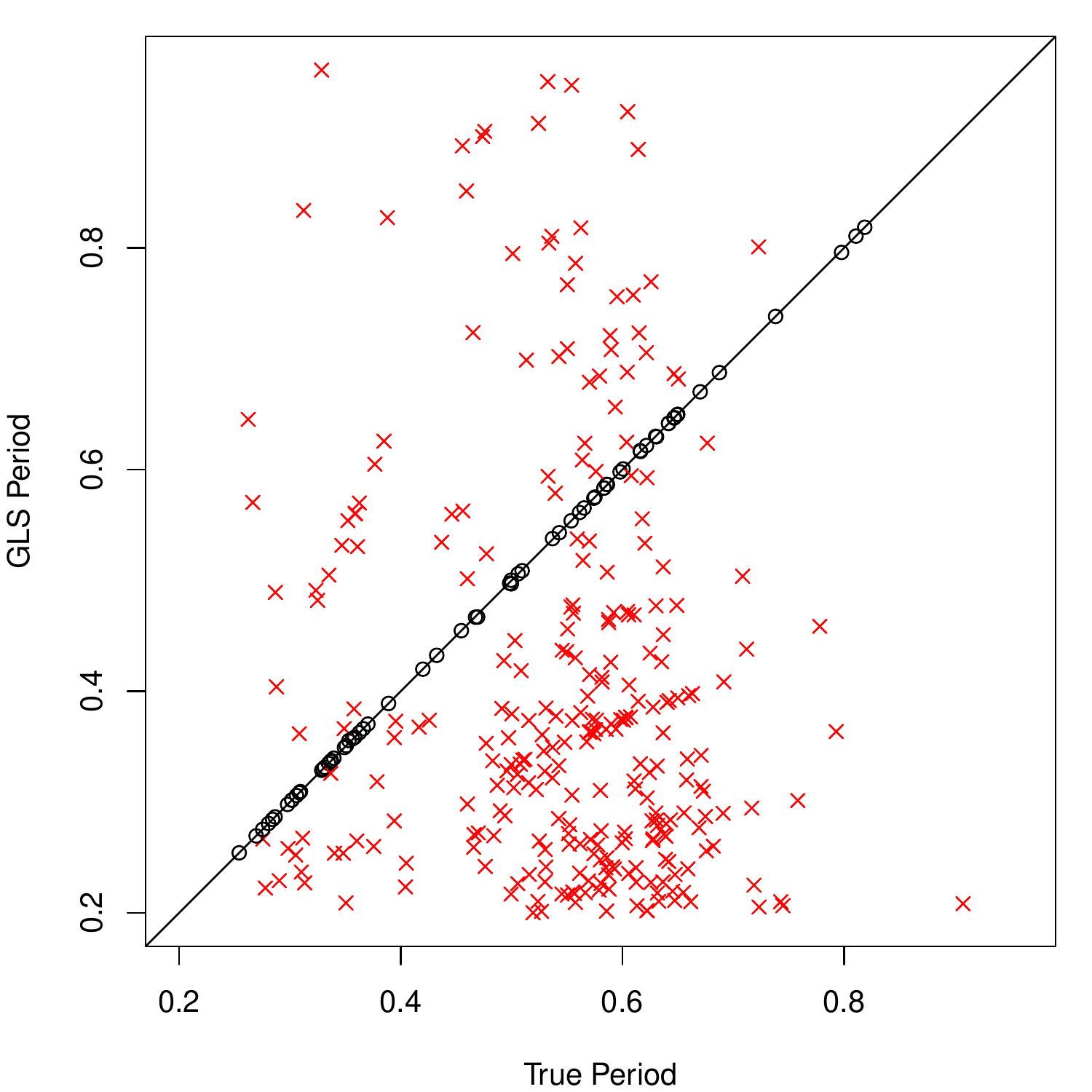}
 } 
 \subfloat[PGLS with 5 points.]{\label{fig:sdss_pvsp5pgls}
\includegraphics[scale=.4]{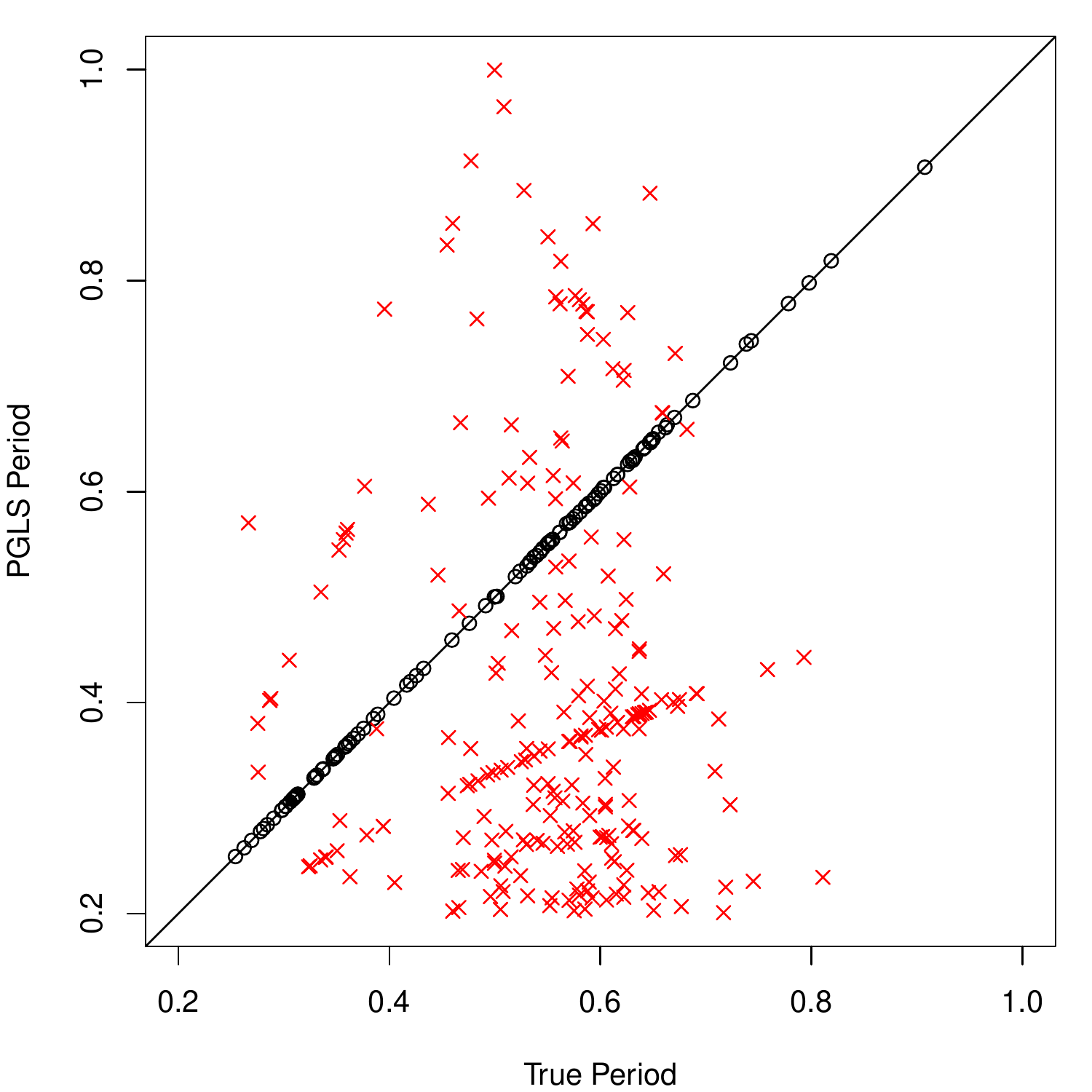}
}\\

 \subfloat[MGLS with 10 points.]{\label{fig:sdss_pvsp10gls}
\includegraphics[scale=.4]{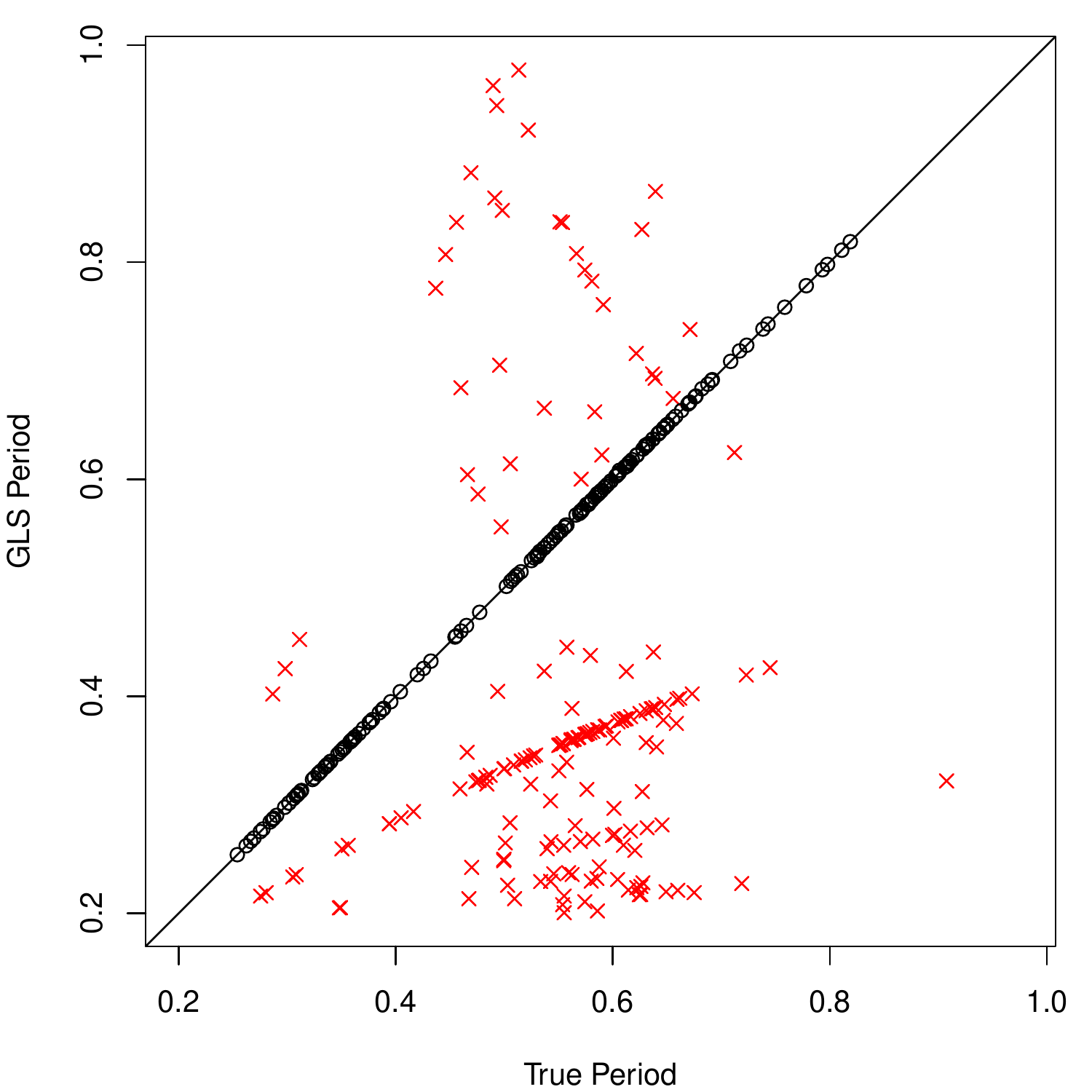}
 } 
 \subfloat[PGLS with 10 points.]{\label{fig:sdss_pvsp10pgls}
\includegraphics[scale=.4]{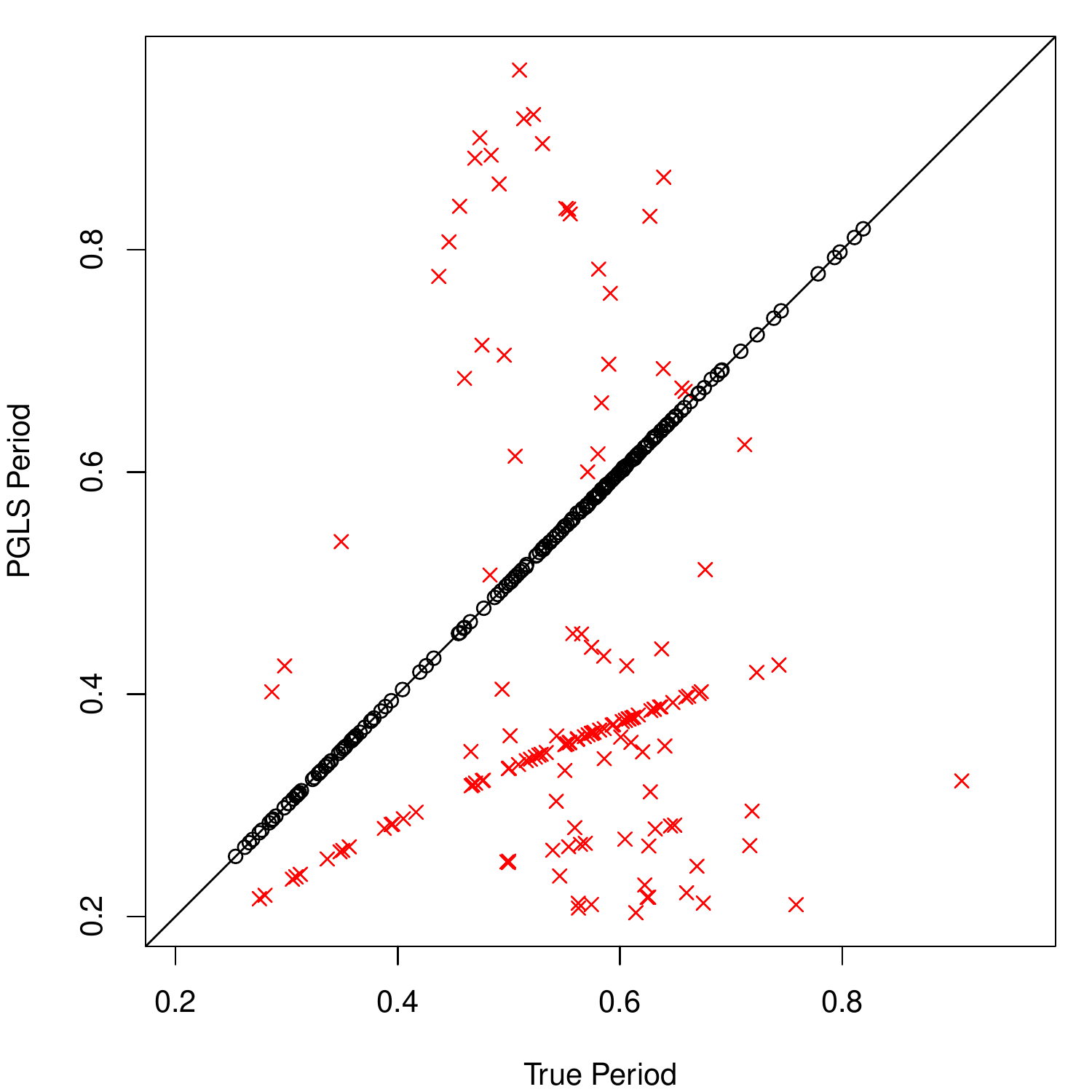}
}\\
 \subfloat[MGLS with 15 points.]{\label{fig:sdss_pvsp15gls}
\includegraphics[scale=.4]{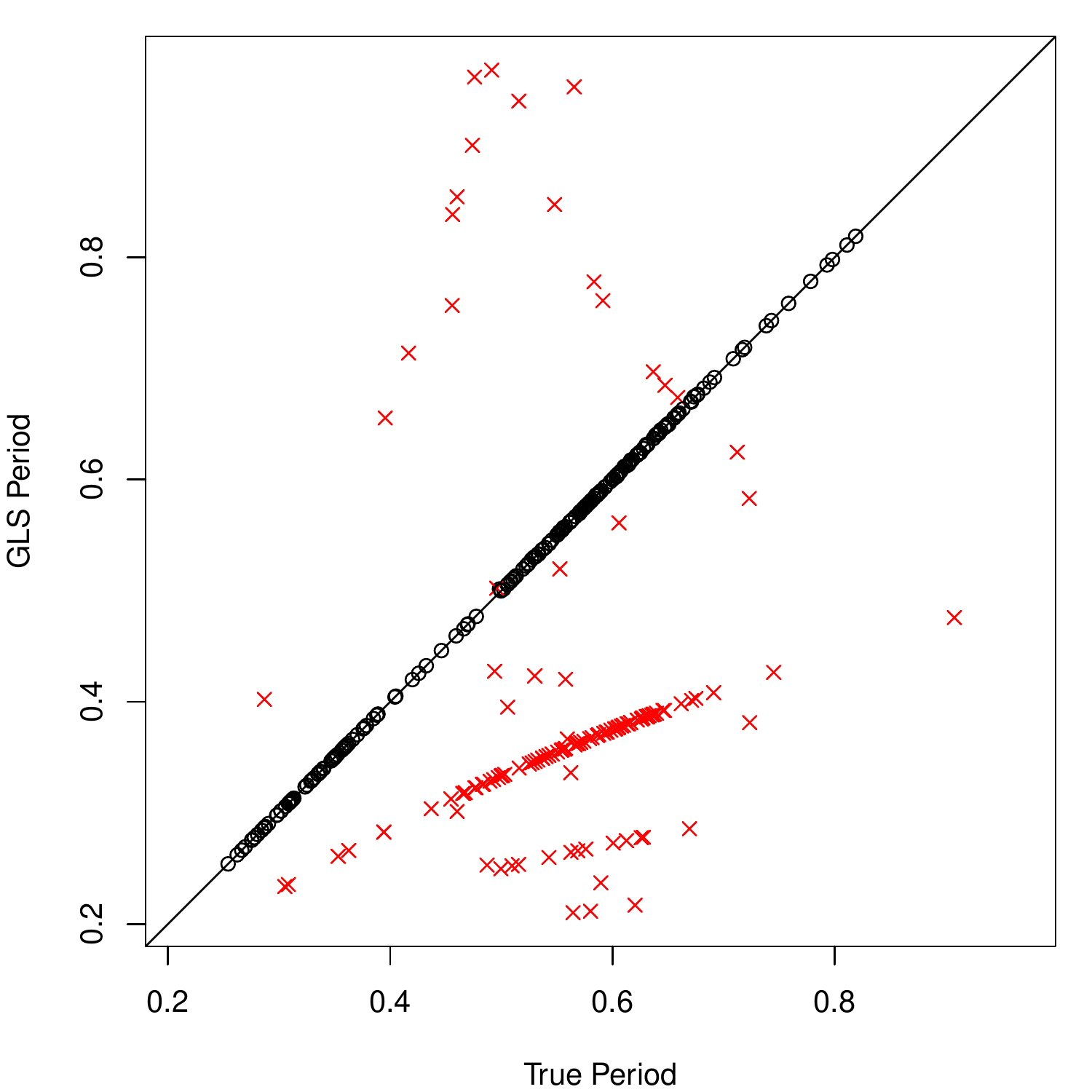}
 } 
 \subfloat[PGLS with 15 points.]{\label{fig:sdss_pvsp15pgls}
\includegraphics[scale=.4]{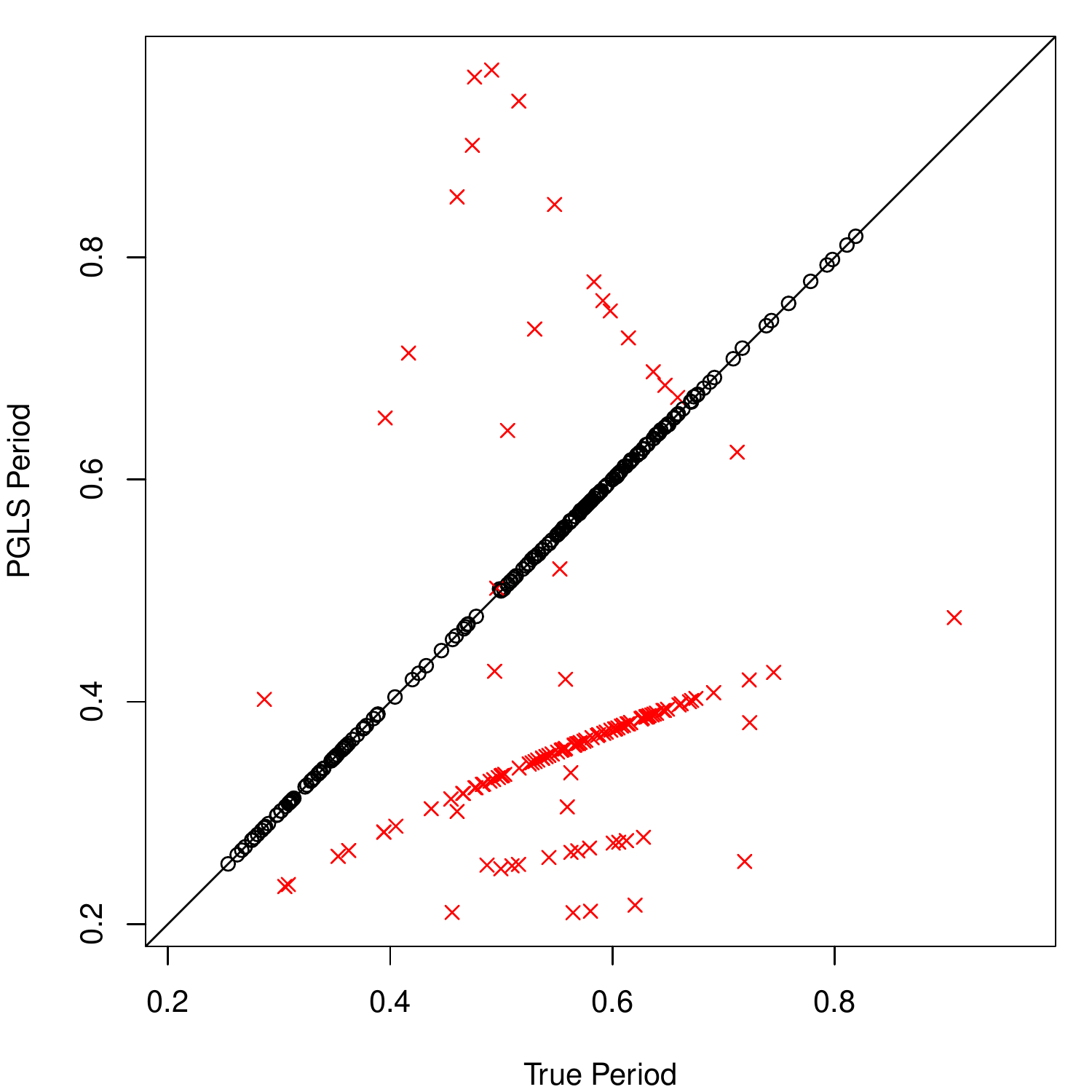}
}
 \caption{Scatterplots of true period versus estimate for MGLS and PGLS with 5, 10, and 15 observations per band for the test set.\label{fig:sdss_accuracies}}
\end{figure}
Table \ref{tab:sdss_tab} shows the fraction of period estimates within 1\% of their true value of MGLS and PGLS. PGLS increases the fraction of periods estimated correctly for 5 and 10 observations per band. Scatterplots of true period versus period estimate for MGLS and PGLS are contained in Figure \ref{fig:sdss_accuracies}. In general, the improvement of PGLS over MGLS appears less for the SDSS data here than the simulated data.

The scatterplots in Figure \ref{fig:sdss_accuracies} suggest that both MGLS and PGLS are not converging to the truth for some stars. In particular, Figures \ref{fig:sdss_pvsp15gls} and \ref{fig:sdss_pvsp15pgls} show that period estimates are either very nearly correct (on the identity line) or some non--linear function of the true period (visible below the identity line). The reason for this phenomenon is likely that RR Lyrae stars exhibit some degree of non--sinusoidality, causing the MGLS period estimate to not be consistent for certain stars. Since PGLS estimates become MGLS estimates as the number of observations increases (and the likelihood washes out the penalty), PGLS is also not consistent for these stars. PGLS appears to increase the rate of convergence of the period estimates to the limiting (in number of observations per band) MGLS value, even when this limit is an incorrect period estimate. For instance, in Figure \ref{fig:sdss_pvsp5pgls} PGLS period estimates concentrate around the non--linear function of the true period more than MGLS estimates in Figure \ref{fig:sdss_pvsp5gls}.

In Figure \ref{fig:sdss_gls_pgls_rho_correlations} we plot phase correlations using MGLS and PGLS for 5 and 15 measurements per band. For both 5 and 15 observations per band, PGLS produces realistic phase estimates concentrated around the identity line. With 15 observations per band the concentration exhibited by MGLS and PGLS is quite close. This is evidence that the penalty terms $J_1$ and $J_2$ are working well.
\begin{figure}[t]
\centering
 \subfloat[5 observations.]{\label{fig:sdss_phases_5}
\includegraphics[scale=.42]{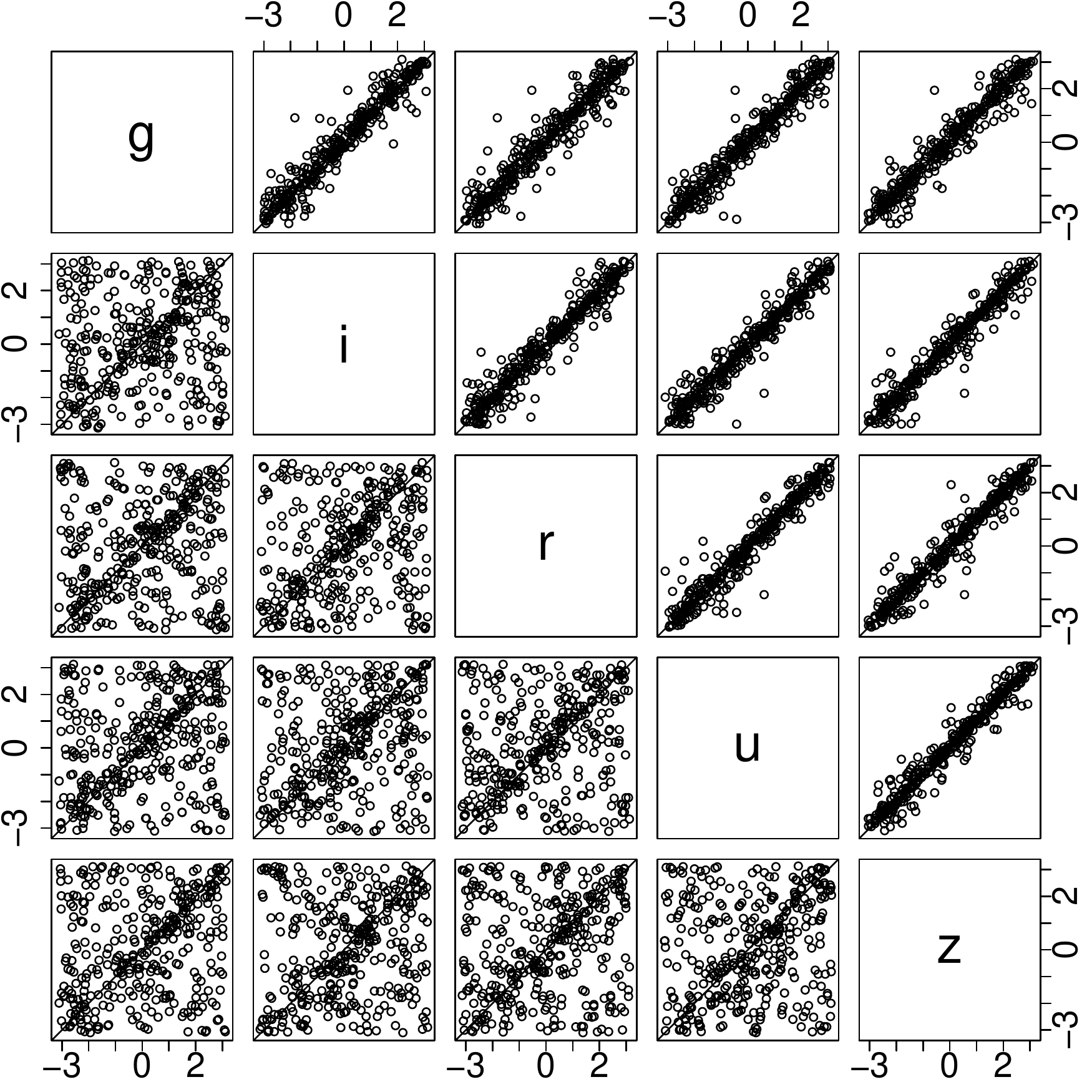}
 } 
 \subfloat[15 observations.]{\label{fig:sdss_phases_15}
\includegraphics[scale=.42]{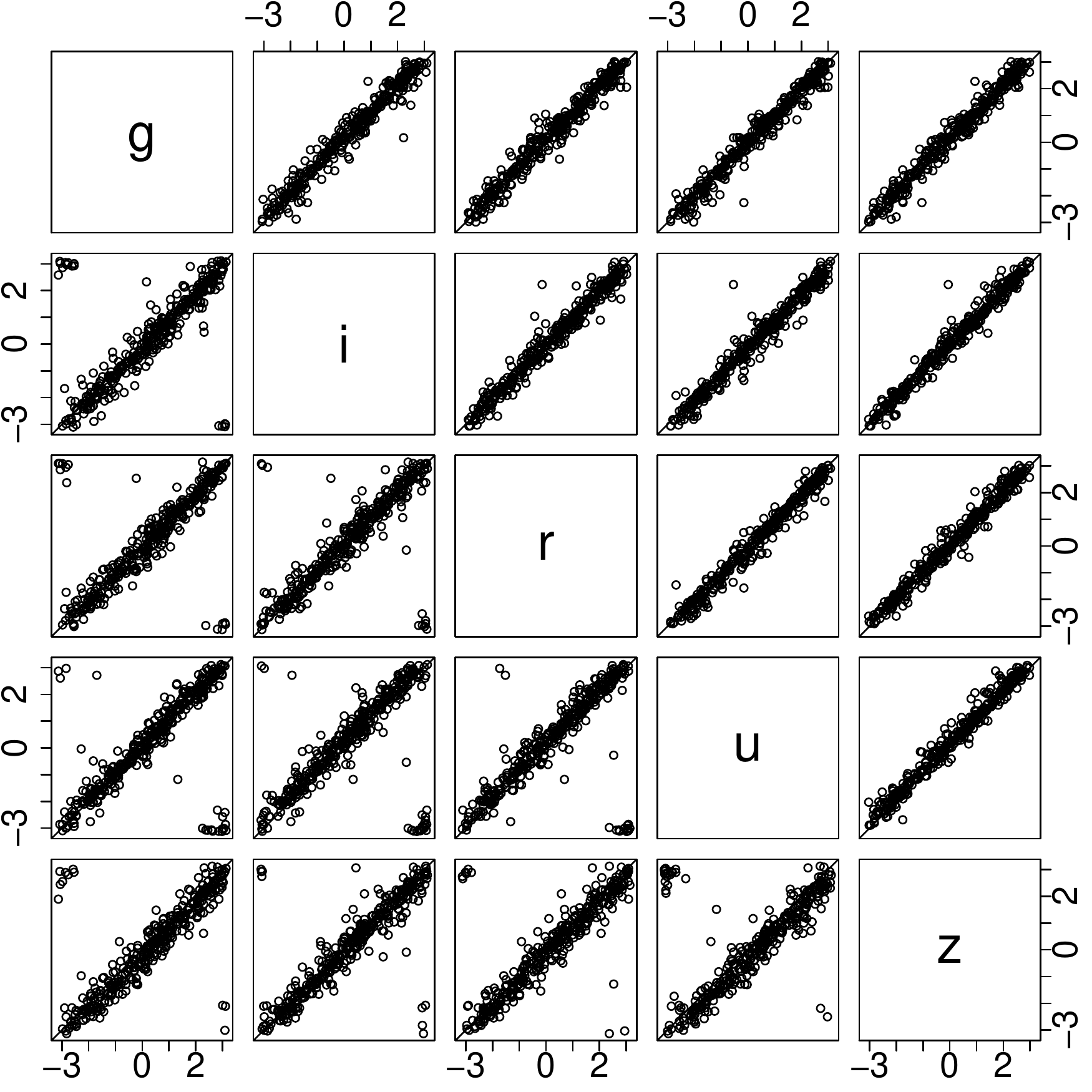}
}
 \caption{Correlations in phase estimates with (a) 5 observations per band and (b) 15 observations per band for the SDSS-II. The results are similar to simulated data. In particular, MGLS period estimates with 5 observations per band show little clustering around $\V{1}$. The PGLS estimates are more physically realistic.\label{fig:sdss_gls_pgls_rho_correlations}}
\end{figure}

\section{Conclusions}
\label{sec:conclusions}

Multiband period estimation is a fundamental problem in modern astronomy. Current methods, however, fail when light curves are poorly--sampled in all measurement bands. To address this deficiency, we introduced two new methods MGLS and PGLS. Both methods generalize a well--known approach to period estimation in astronomy. However, PGLS is the first multiband period estimation algorithm for variable stars that uses known correlations in amplitude and phase across bands to improve accuracy. With simulations and real data, PGLS outperformed MGLS.

Computing the PGLS estimate requires minimizing a more complicated function than computing the MGLS estimate. Nonetheless, the PGLS estimate can be rapidly obtained using our inexact BCD algorithm which requires computational effort that scales linearly with the size of the data. We also showed that even faster PGLS estimates can be obtained by first estimating the computationally cheaper MGLS estimator at different candidate periods and using these estimates to reduce drastically the number of candidate periods on which to compute the PGLS estimate.

We emphasize that both methods, and PGLS in particular, advance the state-of-the-art in multiband period estimation where there are no practical existing alternatives. To the best of our knowledge, even our simple generalization MGLS has not been considered in the astronomy literature. Consequently, we anticipate that further work on our approaches can yield even more accurate estimates.

For example, PGLS was highly effective in simulations where the underlying sinusoidal model is correct. Thus, PGLS will be effective for light curves that are roughly sinusoidal, such as Type II Cepheid variables. In our real data experiments, however, non--sinusoidal behavior in light curve variation decreased period estimation accuracy. This suggests exploring models of light curve variation that are non--sinusoidal.

Furthermore, the penalty terms we developed increased the rate of convergence to the limiting maximum likelihood estimates, suggesting that our penalty terms combined with a non--sinusoidal base model may be effective at estimating periods for poorly--sampled stars that have non--sinusoidal behavior.

The software for implementing MGLS and PGLS can be found in the {\tt multiband} package for R and will be available on CRAN.

\appendix

\section{Majorization in phase update}
\label{sec:mm}

We provide the proof of \Prop{majorization}.
\begin{proof}
Since the function we wish to majorize is Lipschitz differentiable, we can apply the quadratic upper bound principle to construct a convex quadratic majorization \citep{BLin1988}. The gradient and Hessian of $f_b(\rho)$ are given by
\begin{eqnarray*}
f'_b(\rho) 
& = & \VE{a}{b} \langle \VE{a}{b} \V{s}_{b}(\rho) - \V{\mu}_{b}, \M{W}_b\V{c}_{b}(\rho) \rangle
\end{eqnarray*}
and
\begin{eqnarray*}
f''_b(\rho)
& = & \VE{a}{b} \left [\langle \VE{a}{b} \V{c}_{b}(\rho), \M{W}_b\V{c}_{b}(\rho) \rangle + \langle \V{\mu}_{b} - \VE{a}{b} \V{s}_{b}(\rho), \M{W}_b\V{s}_{b}(\rho) \rangle \right].
\end{eqnarray*}
It is straightforward to establish the following global upper bound on the Hessian. Let $\kappa_b = \V{1}\Tra \M{W}_b \V{1}$.
\begin{equation}
\label{eq:bound}
\begin{split}
f''_b(\rho)
\amp \leq \amp &
\VE{a}{b} \left [\VE{a}{b} \V{c}_{b}(\rho)\Tra\M{W}_b\V{c}_{b}(\rho) + \langle \M{W}_b\V{\mu}_{b},\V{s}_{b}(\rho) \rangle \right ] \\
\amp \leq \amp &
\VE{a}{b}  \left [\VE{a}{b} \kappa_b + \langle \M{W}_b\V{\mu}_{b},\V{s}_{b}(\rho) \rangle \right ] \\
\amp \leq \amp &
\VE{a}{b}  \left [\VE{a}{b} \kappa_b + \lVert \M{W}_b\V{\mu}_{b} \rVert_2 \lVert \V{s}_{b}(\rho) \rVert_2 \right ] \\
\amp \leq \amp &
\VE{a}{b}  \left [\VE{a}{b} \kappa_b + \sqrt{n_b}\lVert \M{W}_b\V{\mu}_{b} \rVert_2 \right ].
\end{split}
\end{equation}
The exact second order Taylor expansion of $f_b(\VE{\rho}{b})$ about a point $\tilde{\rho}$ is given by
\begin{eqnarray}
\label{eq:Taylor}
f_b(\rho) & = & f_b(\tilde{\rho}) + f'_b(\tilde{\rho})(\rho - \tilde{\rho}) + \frac{1}{2}f''_b(\rho^\star)(\rho - \tilde{\rho})^2,
\end{eqnarray}
where $\rho^\star \in \alpha \rho + (1-\alpha)\tilde{\rho}$ for some $\alpha \in (0,1)$. The relations (\ref{eq:bound}) and (\ref{eq:Taylor}) together yield the desired result.
\end{proof}

\section{Convergence of the BCD-MM algorithm}
\label{sec:convergence}

We provide the proof of \Prop{convergence}.

\begin{proof}
The convergence theory of monotone algorithms, like BCD and MM, hinge on the properties of the algorithm map $\psi(\V{x})$ that returns the next iterate given the last iterate. For easy reference, we state a simple version of Meyer's monotone convergence theorem \citep{Meyer1976}, which is instrumental in proving convergence in our setting.

\begin{theorem}\label{thm:MM_limit_points}
  Let $f(\V{x})$ be a continuous function on a domain $S$ and
   $\psi(\V{x})$ be a continuous algorithm map from $S$ into $S$ satisfying
 $f(\psi(\V{x})) < f(\V{x})$ for all $\V{x} \in S$ with $\psi(\V{x}) \neq \V{x}$.
  Suppose for some initial point $\V{x}_{0}$ that the set
  $\mathcal{L}_f(\V{x}_{0}) \equiv \{\V{x} \in S : f(\V{x}) \leq f(\V{x}_{0}) \}$ is compact.
Then
  \begin{inparaenum}[(a)]
  \item \label{part:fixed_points} 
    all cluster points are fixed points of $\psi(\V{x})$, and
  \item \label{part:successive_iterates}
     $\lim_{m \to \infty} \lVert \V{x}_{m+1} - \V{x}_{m} \rVert = 0$.
  \end{inparaenum}
\end{theorem}
In the context of \Alg{BCDMM}, the function $f$ is the PNLL \eqref{eq:multiband_problem}. Since \Alg{BCDMM} optimizes over the triple $\V{x} = (\V{\beta_0},\V{a},\V{\rho})$ for a fixed candidate frequency $\omega$, we take the set $S = \Real \times \Real \times [-\pi/2,\pi/2]$ . 

We first use \Thm{MM_limit_points} to establish that the iterates of the inexact BCD algorithm tend towards the fixed points of an algorithm map. 
We then show that the fixed points of the algorithm map correspond to the stationary points of the PNLL.

In order to apply \Thm{MM_limit_points}, we need to i) identify a continuous algorithm map $\psi(\V{x})$ that corresponds to \Alg{BCDMM}, ii) check that $f(\V{x}) > f(\psi(\V{x}))$ if $\V{x} \not = \psi(\V{x})$, and iii) identify an $\V{x}_0 \in S$ so that the set $\mathcal{L}_f(\V{x}_0)$ is compact. 

The first step is to identify an algorithm map $\psi(\V{x})$ and check that it is continuous.
We formalize how to obtain $\V{x}^+ \equiv \psi(\V{x})$ from $\V{x}$ via the composition of three sub-maps, each of which corresponds to a block variable update.
\begin{eqnarray*}
\psi_1( \V{x} ) & = & \argmin{\V{\beta_0}} \ell(\V{\beta_0},\V{a},\V{\rho}), \\
\psi_2( \V{x} ) & = & \argmin{\V{a}} \ell(\V{\beta_0}, \V{a}, \V{\rho}) + \lambda_1 J_1(\V{a}), \\
\psi_3( \V{x} ) & = & \argmin{\V{\rho}'} g(\V{\rho}' \mid \V{\rho}) + \lambda_2 J_2(\V{\rho}').
\end{eqnarray*}
Then the algorithm map that corresponds to the BCD algorithm is
\begin{eqnarray*}
\psi(\V{x}) & = & \begin{pmatrix}
\psi_1(\V{\beta_0},\V{a},\V{\rho}) \\
\psi_2(\psi_1(\V{\beta_0},\V{a},\V{\rho}),\V{a},\V{\rho}) \\
\psi_3(\psi_1(\V{\beta_0},\V{a},\V{\rho}), \psi_2(\psi_1(\V{\beta_0},\V{a},\V{\rho}),\V{a},\V{\rho}),\V{\rho}) \\
\end{pmatrix}.
\end{eqnarray*}
Inspecting the block updates \Eqn{update_beta0}, \Eqn{update_a}, and \Eqn{update_rho_explicit}, we see that the sub-maps $\psi_1, \psi_2,$ and $\psi_3$ are each continuous. Therefore, the composition map $\psi$ is also continuous.

The second step is to verify that $f(\V{x}) > f(\psi(\V{x}))$ if $\V{x} \not = \psi(\V{x})$. Consider the intermediate iterates
\begin{eqnarray*}
\V{\beta_0}^+ & \equiv & \psi_1( (\V{\beta_0}, \V{a}, \V{\rho}) ) \\
\V{a}^+ & \equiv & \psi_2( (\V{\beta_0}^+, \V{a}, \V{\rho}) ) \\
\V{\rho}^+ & \equiv & \psi_3( (\V{\beta_0}^+, \V{a}^+, \V{\rho}) ).
\end{eqnarray*}
For any $\V{x} \in S$, we have
\begin{eqnarray*}
f(\V{x}) \amp \geq \amp f( (\V{\beta_0}^+, \V{a}, \V{\rho})) \amp \geq \amp f((\V{\beta_0}^+, \V{a}^+, \V{\rho})) \amp \geq \amp f(\V{x}^+).
\end{eqnarray*}
If, however, $\V{x}$ is not a fixed point of $\psi$, namely $\psi(\V{x}) \not = \V{x}$, then at least one of the above inequalities is strict and therefore, $f(\V{x}) > f(\psi(\V{x}))$.

The third step is to identify an $\V{x}_0 \in S$ so that the set $\mathcal{L}_f(\V{x}_0)$ is compact. Note that $f$ blows up for $\lVert \V{A} \rVert \rightarrow \infty$ or $\lVert \V{\beta_0} \rVert \rightarrow \infty$. Therefore, the set $\mathcal{L}_f(\V{x}_{0})$ is compact for any initial $\V{x}_0 \in S$.

By \Thm{MM_limit_points}, it follows that the cluster points of the iterates of \Alg{BCDMM} are fixed points of the algorithm map $\psi(\V{x})$ corresponding to \Alg{BCDMM}.
To complete the proof we just need to show that every fixed point of $\psi(\V{x})$ is a stationary points of the PNLL \eqref{eq:multiband_problem}.
If $\V{x}$ is a fixed point of $\psi$ then,
\begin{eqnarray*}
\frac{\partial}{\partial \V{\beta_0}}f(\V{x}) & = & \V{0}, \\
\frac{\partial}{\partial \V{a}}f(\V{x}) & = & \V{0}, \\
\frac{\partial}{\partial \V{\rho}} \left [g(\V{\rho} \mid \V{\rho}) + \lambda_2 J_2(\V{\rho}) \right ] & = & \V{0}. \\
\end{eqnarray*}
In light of \Eqn{maj_stationarity} it is clear that the last condition is equivalent to
\begin{eqnarray*}
\frac{\partial}{\partial \V{\rho}} f(\V{x}) & = & \V{0}.
\end{eqnarray*}
Therefore, every cluster point of the iterate sequence generated by \Alg{BCDMM} is a stationary point of the PNLL \eqref{eq:multiband_problem}.
\end{proof}

%

\bibliographystyle{imsart-nameyear}
\bibliography{refs}

\end{document}

%% file: figs/cep_gls_tab.tex
36.5\%

%% file: figs/cep_pgls_tab.tex
31.5\%

%% file: figs/sdss_results_output_results.tex
\begin{table}[t]
\centering
\caption{Fraction of period estimates within 1\% of truth.\label{tab:sdss_tab}} 
\begin{tabular}{rrr}
  \hline
obs. / band & GLS & PGLS \\ 
  \hline
  5 & 0.19 & 0.34 \\ 
   10 & 0.51 & 0.59 \\ 
   15 & 0.64 & 0.64 \\ 
   \hline
\end{tabular}
\end{table}